\documentclass[twocolumn,10pt]{IEEEtran}
\usepackage{bm,cite,float,amsmath,amssymb,amsthm}

\newcommand*\diff{\mathop{}\!\mathrm{d}}

\usepackage{algorithm}
\usepackage[noend]{algpseudocode}
\usepackage{indentfirst}

\makeatletter
\def\BState{\State\hskip-\ALG@thistlm}
\makeatother

\DeclareMathOperator\erf{erf}
\DeclareMathOperator\erfc{erfc}

\usepackage{amssymb}
\usepackage{amsmath}
\usepackage{graphicx}
\usepackage{cite}
\usepackage{citesort}
\usepackage{balance}
\usepackage[utf8]{inputenc}

\usepackage{graphicx, epstopdf}
\usepackage{setspace}
\usepackage{acronym}
\usepackage{mathrsfs}
\usepackage{ifthen}
\usepackage{textcomp}
\usepackage{float}
\usepackage{subfigure}
\usepackage{color}
\usepackage{graphicx, epstopdf}
\usepackage{setspace}
\usepackage{graphicx,cite,amsmath,amssymb,hhline}
\usepackage{lineno}

\usepackage{array}
\usepackage{mdwmath}
\usepackage{mdwtab}
\usepackage{verbatim}
\usepackage{color}
\usepackage{multirow}

\newtheorem{lemma}{Lemma}
\newtheorem{cor}{Corollary}
\newtheorem{ppro}{Proposition}

\IEEEoverridecommandlockouts

\usepackage{graphicx,epstopdf}
\usepackage{epsfig}	
\usepackage{amsfonts,balance}
\usepackage{bbm}
\floatname{algorithm}{Algorithm}
\usepackage{lipsum}

\newtheorem{theorem}{Theorem}

\newtheorem{corollary}{Corollary}

\usepackage[
top    = 1.70cm,
bottom = 1.05in,
left   = 0.63 in,
right  = 0.63 in]{geometry}

\makeatletter
\def\ScaleIfNeeded{%
\ifdim\Gin@nat@width>\linewidth \linewidth \else \Gin@nat@width
\fi } \makeatother

\begin{document}

\title{    Analyzing  Large-Scale  Multiuser Molecular Communication   via  3D   Stochastic Geometry   }

%\author{
%\IEEEauthorblockN{ Yansha Deng\IEEEauthorrefmark{1}, Adam~Noel\IEEEauthorrefmark{2}, Weisi~Guo\IEEEauthorrefmark{3}, 
%Arumugam~Nallanathan\IEEEauthorrefmark{1}, 
%and  Maged~Elkashlan\IEEEauthorrefmark{4}  
%%and
%%Robert Schober\IEEEauthorrefmark{2}\IEEEauthorrefmark{4}
% } \IEEEauthorblockA{
%\IEEEauthorrefmark{1}Department of Informatics, King's College London,  UK\\
%\IEEEauthorrefmark{2} School of Electrical Engineering and Computer Science, University of Ottawa, Canada \\
%\IEEEauthorrefmark{3} School of Engineering, University of Warwick, UK\\
%\IEEEauthorrefmark{4} School of Electronic Engineering and Computer
%Science, Queen Mary University of London, UK\\
%%\IEEEauthorrefmark{4}Institute for Digital Communication, Friedrich-Alexander-Universit\"{a}t Erlangen-N\"{u}rnberg (FAU), Erlangen, Germany\\
% } }

\author{Yansha~Deng,~\IEEEmembership{Member,~IEEE,}
        Adam~Noel,~\IEEEmembership{Member,~IEEE,}
         Weisi~Guo,~\IEEEmembership{Member,~IEEE,}\\
        Arumugam~Nallanathan,~\IEEEmembership{Fellow,~IEEE,}
        and~Maged~Elkashlan,~\IEEEmembership{Member,~IEEE}.

\vspace{-0.2cm}
\thanks{ Manuscript received Sep. 13, 2016; revised May. 9, 2017 and  Aug. 28, 2017; accepted Aug. 30, 2017.  This paper was presented in part at the Proc. IEEE Global Telecommunications Conf. (GLOBECOM’16), Washington D.C., USA, Dec. 2016 \cite{yansha2016stochastic}. The  editor coordinating the review of this
manuscript and approving it for publication was Prof. Melodia, Tommaso.}
\thanks{
Y. Deng and A. Nallanathan are with  Department of Informatics, King's College London, London, WC2R 2LS, UK (email:\{yansha.deng, arumugam.nallanathan\}@kcl.ac.uk).}
\thanks{A. Noel is with the School of Electrical Engineering and Computer Science, University of Ottawa, Ottawa, ON, K1N 6N5, Canada (email: anoel2@uottawa.ca).}
\thanks{ Weisi~Guo is with School of Engineering, University of Warwick, West Midlands, CV4 7AL, UK  (email:  weisi.guo@warwick.ac.uk).}
\thanks{
 A. Nallanathan and M. Elkashlan are with Queen Mary University of London, London, E1 4NS,
UK  (email:\{arumugam.nallanathan, maged.elkashlan\}@qmul.ac.uk).}
}

%{$^{\ast}$Department of Computer Science and Operations Research\\
%	Universit\'{e} de Montr\'{e}al\\
%	Email: adamnoel@umontreal.ca
%	\\ $^{\dagger}$School of Electrical Engineering and Computer Science\\
%	University of Ottawa}

\newcommand{\meter}{\textnormal{m}}
\newcommand{\micron}{\mu\textnormal{m}}
\newcommand{\second}{\textnormal{s}}

% Changes in newer revision
%\newcommand{\edit}[2]{\textbf{#1}}
\newcommand{\edit}[2]{#1}

\maketitle
\vspace{-3cm}

\begin{abstract}
Information delivery using chemical molecules is an integral part of biology at multiple distance scales and has attracted recent interest in bioengineering and communication theory. Potential applications include cooperative networks with a large number of simple devices that could be randomly located (e.g., due to mobility). This paper presents the first tractable analytical model for the collective signal strength due to randomly-placed transmitters in a three-dimensional (3D) large-scale molecular communication system, either with or without degradation in the propagation environment. { Transmitter locations in an unbounded and homogeneous fluid are modelled as a homogeneous Poisson point process.} By applying stochastic geometry, analytical expressions are derived for the expected number of molecules absorbed by a fully-absorbing receiver or observed by a passive receiver. The bit error probability is derived under ON/OFF keying and either a constant or adaptive decision threshold. Results reveal that the combined signal strength increases proportionately with the transmitter density, and the minimum bit error probability can be improved by introducing molecule degradation. Furthermore, the analysis of the system can be generalized to other receiver designs and other performance characteristics in large-scale molecular communication systems.

%In this report, we model the chemical signal strength from $K$ molecular messenger nodes.  With the aid of stochastic geometry, we show that if the messenger nodes' location follow a uniform random spatial distribution, the received signal strength statistics from one or more nodes can be represented by closed form expressions.  We then prove that in the absence of individual synchronisation between a pair of messengers, random asynchronous transmission is superior to all the messengers sharing a common preset transmission clock.
\end{abstract}


\section{Introduction}

Molecular communication via diffusion has attracted significant bioengineering and communication engineering research interest in recent years \cite{EckfordBook13}.  Messages are delivered via molecules undergoing random walks \cite{Codling08}, which is a prevalent phenomenon in biological systems and between organisms \cite{Atkinson09} across multiple distance scales, offering transmit energy and signal propagation advantages over wave-based communications \cite{SMIET2017,Llatser13,Guo15TMBMC}. More importantly,  when compared to { electromagnetic wave-based communication systems}, molecular communication can be advantageous at very small dimensions or in specific environments, such as in salt water or human bodies.

Fundamentally, molecular communications involves modulating information on the physical properties (e.g.,  number, type, emission time) of a single molecule or group of molecules  (such as pheromones, DNA, protein). 
When modulating  the number of  molecules, each messenger node will transmit information-bearing molecules via chemical pulses. According to the theory of  Brownian motion, the average displacement of each molecule is proportional to its diffusion time and  the diffusion coefficient, however, the instantaneous displacement of each molecule differs and  is usually described by the Normal distribution \cite{howard1993radom630}. As such, a molecule emitted in a previous bit interval may arrive at the receiver during the current interval, thereby confusing  the signal detection at the receiver with intersymbol interference (ISI).

 Existing works have largely focused on  modeling   the signal strength  of a point-to-point communication channel by taking into account the self-interference that arises from previous symbols (i.e., ISI) at a passive receiver \cite{Kuran10}, at a fully absorbing receiver \cite{yilmaz2014simulation}, and at a reversible adsorption receiver \cite{Yansha16}.  
 Efforts to mitigate ISI include %replacing the absorbing receiver with the adsorption$\&$desorption receiver \cite{Yansha16},  
 transmitting using two different types of molecules in consecutive bit intervals \cite{tepekule2015ISI},  and designing  ISI-free codes \cite{shih2013code}.

Recent advances in bio-nanotechnology bring new opportunities for enabling molecular communication in new applications, such as  drug delivery, environmental monitoring, and pollution control. One application example is that  swarms of nano-robots could track specific targets, such as  tumour cells, to perform operations such as targeted drug delivery \cite{Douglas12}.  In such a scenario, each nano-robot may receive the signal transmitted from multiple nano-robots.  Thus, how to establish
energy efficient and tether-less communication    becomes an important research problem \cite{Cavalcanti06}.

%External molecule sources (external interference) are the other contributor to the reliability degradation when nanonetworks are deployed. The external molecule sources can  result from: 1) the multiuser interference due to that a swarm of transmitters  (e.g., robots or nanomachines)
% operating together  to transmit molecular messages simultaneously using the same type of information molecules \cite{noel2014unify}; 2) the  product of potential chemical reactions of the other type of information molecule  happens  in   the  diffusive environment \cite{alberts2013essential}; and  3) the accidental leakage from the nanomachines loaded with the same type of information molecules \cite{ladokhin1995leakage}.

In   nanonetworks, it is therefore important to provide a physical model for the  collective signal strength at the receiver in a large-scale system, while taking into  account  random transmitter locations due to mobility.   In  \cite{adam2014unify}, the collective signal strength  of a multi-access communication channel at a passive receiver due to co-channel transmitters (i.e., transmitters emitting the same type of molecule) was measured given the knowledge
 of their total  number and locations. In \cite{multiplea2015amini}, the capacity of the multiple access channel  with a single bit emitted at each transmitter and  a ligand-binding receiver was derived  under the assumption of a deterministic diffusion channel model.
{ The first work to consider randomly distributed co-channel transmitters in  a 3D
diffusion channel according to a spatial { homogeneous Poisson  point process (HPPP)}
is \cite{pierobon2014statistical}, where the probability density function (PDF) of the received power spectral density at a point receiver was derived based on the assumption of
white Gaussian transmit signals. The analysis in \cite{pierobon2014statistical} { considers} multiuser
emission within a single transmission interval, and the presented results {are} Monte Carlo simulations. }

% 
% Due to limitations in transmitter design and molecule type availability, it is likely that many transmitters will transmit the same type of information molecule.
%Thus,  it is important to model the  collective signal strength due to all transmitters with the same type of information molecule, and to account for random transmitter locations due to mobility. 

%\cite{Wyatt09}
% Examples include a reverse-engineered pheromone signalling system \cite{Cole09} and a macro-scale prototype that can reliably send generic text messages using molecule \cite{Farsad13PLOS}.

From the perspective of receiver design, many works have focused on the passive receiver, which can observe and count the number of molecules inside the receiver without interfering with the molecules \cite{Kuran10,adam2014unify,pierobon2014statistical}. In nature, receivers commonly remove information molecules from the environment once they bind to a receptor. An ideal model is the fully absorbing receiver, which absorbs all the molecules hitting its surface \cite{yilmaz2014simulation,Yansha16}. 
Unfortunately, no work has studied the channel characteristics and the received signal at a fully absorbing receiver in a large-scale  molecular communication system, nor compared it with 
that at a passive receiver.

In this paper, we aim to provide an analytical model and bit error probability for the collective signal  
 at the passive receiver and the fully absorbing receiver  due to  a swarm of active mobile point transmitters that simultaneously emit the same  bit sequence. 
{  We extend our previous work in \cite{yansha2016stochastic} by deriving the bit error probability of a constant threshold detector at both receivers under fixed threshold-based demodulation, and applying decision feedback detection (DFD) for performance improvement.
   Our new analysis takes into account the molecule degradation
during diffusion based on the following three facts: 1)  molecules are unlikely to persist for all time, and may be degraded by  chemical reactions in
a biological environment; 2) the constant transmitter density over unbounded
space assumed in our analytical model implies that there is an infinite number of
transmitters and that ISI increasingly accumulates, which isn't practical; 3) the molecule degradation will help to reduce the
ISI and improve the probability of error.}
 
% Differently and unlike previous contributions on the expected observations at the receivers in \cite{yansha2016stochastic}, here,  the molecule degradation  is introduced  during diffusion to remove molecules released by distant transmitters before they eventually arrive as ISI, and the bit error probbilities based on the instantaneous observations are formulated and derived.
 The analytical results are obtained via the powerful tools of    stochastic geometry in 3D space, which can characterize the average behavior over many spatial realizations of a network where the transmitter nodes are placed according to some probability distribution \cite{Haenggi12}. Just as we can analyze the network performance of a random field of transmitters in conventional wireless networks, we can also apply a similar rationale for analyzing  the receiver performance due to a swarm of molecular transceivers. { However, unlike \cite{hasan2007guard,novlan2013uplink}, where the network performance is analyzed based on the distribution of the received signal-to-interference-plus-noise ratio (SINR) at a
point receiver in 2D space, we seek the mean and distribution of the number
of received molecules at a  \emph{spherical} receiver due to all  transmitters in 3D
space.}
By doing so, simple and tractable results can be obtained to reveal the key dependency of the molecular communication system performance metrics with respect to the system parameters. 
This work and \cite{yansha2016stochastic} 
  are distinct from  related work in \cite{pierobon2014statistical}, which  focused on the  statistics of the received signal at any \emph{point} location. 
Our contributions can be summarized as follows:
\begin{enumerate}
\item Using  stochastic geometry, we model the collective signal at a receiver in a 3D large-scale  molecular communication system with or without molecule degradation, where the receiver is either passive or fully absorbing. To examine the impact of the signal from the nearest transmitter relative to the aggregate signal, we also derive the signals from the nearest transmitter and the other transmitters.
%We distinguish between the desired signal   due to the nearest transmitter and the interfering  signal due to the other   transmitters.
\item We derive a general expression for the expected net number of 
molecules observed at both types of receivers during any time interval. In order to gain insights about the impact of the transmitter density, the diffusion coefficient, and the receiver radius on the collective signal, we simplify the general expression to a \emph{closed-form} expression 
for the expected net
number of  molecules absorbed at the fully absorbing
receiver under molecule degradation.
%, in order to gain insights for the impact of the transmitter density, the diffusion coefficient, and the receiver radius on the collective signal.
%\item We  define and derive a tractable analytical expressions for the signal-to-interference ratio (SIR) at the  receiver, where the signal from the nearest transmitter is the useful signal and signals from the other transmitters are interference.
\item  We derive a general expression for the bit error probability at the passive or  absorbing receiver in the proposed  system   with or without molecule degradation under ON/OFF keying.  A simple  detector requiring one sample per bit interval is considered as a preliminary design for the proposed large-scale system.  Importantly, this general expression for the bit error probability can also be applied for  other types of receivers by substituting the  corresponding channel response.

\item { We focus on Monte Carlo simulation approaches to verify our analytical results, and we also compare Monte Carlo simulation to particle-based simulation of the large-scale molecular communication system.} It is shown 
that the  expected number of molecules observed at both types of receivers increases linearly with increasing transmitter density. We also show that the minimum bit error probability of both receivers can be improved by introducing  molecule degradation. 
\end{enumerate}

The rest of this paper is organized as follows. In Section II, we introduce the system model. In Section III, we present the channel impulse response of information molecules  in the large-scale molecular communication system.  In Section IV, we derive
 the exact and asymptotic net  number of absorbed molecules
expected at the surface of the absorbing receiver, and the exact  number of molecules observed inside the passive receiver in  the large-scale molecular communication system. 
 In Section V, we derive the bit error probability of the proposed system with a  simple detector requiring one sample per bit.  In Section VI, we present the numerical and simulation results. In Section VII, we conclude the contributions of this paper.

\section{System Model}

In Fig. \ref{System}, we consider a  3D diffusion-based molecular communication system with a single receiver located at the origin under joint transmission by a swarm of point transmitters, which are  spatially distributed outside the receiver in  $\mathbb{R}^3/ {V_{{\Omega _{{r_r}}}}}$ according to an independent and homogeneous Poisson point process (HPPP) \footnote{{This model is also valid for  spherical transmitters  with transparent membranes, where the locations of the point process are the molecule emission points.}}  $\Phi_a$ with density $\lambda_a$, where ${V_{{\Omega _{{r_r}}}}}$ is the volume of receiver ${{\Omega _{{r_r}}}}$ with ${V_{{\Omega _{{r_r}}}}} = {{4\pi r_r^3} \mathord{\left/
 {\vphantom {{4\pi r_r^3} 3}} \right.
 \kern-\nulldelimiterspace} 3}$.  
 HPPP has been widely used to model wireless sensor networks \cite{hasan2007guard,yan2016sensor}, homogeneous and heterogeneous cellular networks \cite{Haenggi12,yan16hetnet},   and has also been applied to model  bacterial colonies in \cite{Jeanson11} and the  interference sources in a  molecular communication system \cite{pierobon2014statistical}. 
 {  We note that we  focus on an unbounded fluid environment  with uniform diffusion and no flow currents to provide a  baseline  for the design of more complicated scenarios in  future works.}
 
%   At any given time instant, a number of point transmitters will be either silent or active. When the transmitters are active, they will cooperatively transmit the same bit sequence  to the receiver. Thus, we define the activity probability of a point transmitter that is triggered to transmit data  as ${\rho_a}$ $(0<{\rho_a}<1)$. This activity probability  is independent of the receiver's location and time. According to the thinning theorem that the independent thinning of a PPP is again a PPP, thus, the active point transmitters constitute independent HPPPs $\Phi_{a}$ with intensities $\lambda_a =\lambda \rho_{a}$ \cite{Haenggi12}.
   
 \begin{figure}[t]
	\centering
	\includegraphics[width=1.00\linewidth]{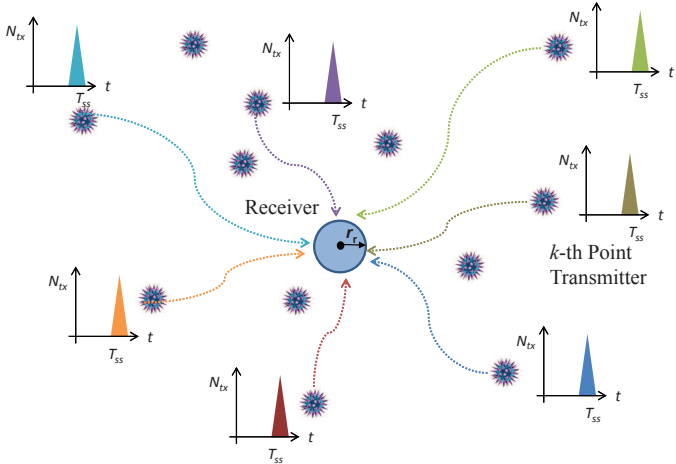}
	\caption{Illustration of a receiver receiving molecular pulse signals from point transmitters at different distances.}
	\label{System}
\end{figure}

We consider a spherical receiver that is either passive \cite{Llatser13,noel2014improving} or fully absorbing \cite{Yilmaz14}.
%   	    The single spherical receiver can either be the fully absorbing receiver \cite{Yilmaz14}, or the  passive receiver \cite{Llatser13,noel2014improving}. 	
   The fully absorbing receiver is covered with selective independent receptors, which are only sensitive to a single type of  information molecule. Similar to \cite{Yilmaz14,Yansha16}, we   assume  that there is no physical limitation on the number of receptors on the surface of the receiver, which is an  appropriate assumption   for a system with a sufficiently large number of receptors or small number of absorbed molecules.   Any information molecule that diffuses into the sphere is absorbed by  a receptor and counted for information demodulation.
   The passive receiver is covered  with a transparent membrane that is permeable to the information molecules passing by, and  the number of information molecules inside the receiver can be counted for information demodulation as in \cite{adam2014unify}.
 
   Even though this work considers a molecular communication system with a single receiver, it provides fundamental insights that can be applied to consider systems with multiple  transceivers  in future work. For example, in the case of multiple passive receivers following HPPP,  the expected number of molecules inside  any passive receiver with the same radius will be equivalent due to  the Slivnyak-Mecke's theorem \cite{baccelli2009stochastic}. In other words, the presence of multiple passive receivers will not influence the observations at each passive receiver,  due to its transparent membrane.   
   This is also consistent with the stochastic geometry work on cellular networks, where the average ergodic rate of an arbitrary random mobile user is expressed using a single expression \cite{novlan2013uplink}.
   However, for the case of multiple absorbing receivers, the  numbers of molecules absorbed by  each absorbing receiver is not independent; in other words,  the presence or absence of an absorbing receiver influences  the  numbers of molecules absorbed by other  absorbing receivers. For a system of absorbing receivers, 
    the average observation will be harder to characterize, but understanding the single receiver system   is still the first step.
    
    A molecular communication system typically includes five processes: emission, propagation, reception, modulation, and demodulation, which are presented in detail  in the following subsections for the absorbing receiver and the passive receiver, respectively.
    \subsection{Emission $\&$ Modulation}
    Applying ON/OFF keying as in \cite{Yilmaz14,Yansha16},
    each transmitter  delivers molecular signal pulses with  $N_{\rm{tx}}$  type $S$ information molecules to the  receiver   at the start of each bit interval to represent transmit bit-1, and emits zero molecules to deliver bit-0.
 { Here, a global clock is assumed at each transmitter such that  the molecule  emissions at all the transmitters are synchronized with the same bit sequences\footnote{{ One application  is that nanomachines could send the same molecular
signal upon sensing some threshold value in the environment \cite{pierobon2014statistical}.
Perfect synchronization between all transmitters is an idealization that facilitates the analysis and leads to tractable results. However, it is not essential for the accuracy of our results, since the distribution in molecule arrival times is primarily determined by the transmitter locations.}},  and  can only occur at the start of a bit interval as in \cite{pierobon2014statistical}.  Asynchronous emission can be evaluated similarly
to synchronous emission by allowing transmitters to release molecules at the
start of intervals that are much smaller than the bit interval, and scaling the
transmitter density accordingly.}
 
 %We assume the center of the absorbing receiver and passive receiver are both located at the origin, and their radius are $r_r$.
 \subsubsection{Absorbing Receiver}
 In the absorbing receiver scenario, we assume  spherical symmetry, where
the transmitter is \emph{effectively} a point on the  spherical shell with radius $r_0$ away from the center of receiver and the molecules
are released from random points over the shell at $t=0$; the actual
angle to the transmitter when a molecule hits the receiver  is
irrelevant. 
{Thus, we define the initial condition as \cite[Eq.~(3.61)]{schulten2000lectures}
 \begin{align}
C^{\text{FA}}\left( {r,\left. {t \to 0} \right|{r_0}} \right) = \frac{1}{{4\pi {r_0}^2}}\delta \left( {r - {r_0}} \right), \label{initial1}
\end{align}
where $C^{\text{FA}}\left( {r,\left. {t \to 0} \right|{r_0}} \right)$  is the molecule distribution function at time ${t \to 0}$ and distance $r$ with initial distance $r_0$ .}

{ According to \eqref{initial1},  there	 is	 spherical symmetry	 that	 makes	 the	 molecules initially	
distributed	 	with	 equal	 probability	 over	 a	 spherical	 surface	 at	 distance	 $r_0$	 from	 the	 receiver.	
Mathematically,	 Eq.	\eqref{initial1}  represents the impulse response  averaged over the surface area of the ball, where ${4\pi {r_0}^2}$
is the surface area of the ball centered at the center of receiver. The	 direct	 interpretation	 	 is	 that,	due to spherical symmetry, a shell transmitter is analogous to a point transmitter.	 
As an example, consider an absorbing receiver 	and	two	molecules	that are	initially	
placed	at	two	points equidistant from the receiver.	Each	molecule	has	the	same	probabilistic	trajectory	for	hitting	the	
receiver.	 Since	 they	 are	equidistant	 from	 the	 receiver,	 we can merge	 them	 to	 a	 single	 point	
source	 to	 achieve	 the	 same	 result.}

 \subsubsection{Passive receiver} 
{ In the passive receiver scenario, we assume an asymmetric spherical model,
which accounts for the actual angle of the molecule inside the passive receiver.
  The information particles are injected into the fluid environment by a transmitter located at   ${\overrightarrow{r}}$ away from  the center of the passive receiver \cite{howard1993radom630}.}
    
    \subsection{Diffusion Under Molecule Degradation}
   
   The diffusion of molecules in the propagation process  follows   random  Brownian motion. With a sufficiently low concentration of information molecules in the fluid environment, the collisions between these molecules can be ignored and the molecules propagate independently with constant diffusion coefficient\footnote{ The diffusion coefficient can be obtained via experiment or estimated via the
Stokes-Einstein equation for spherical molecules  \cite[Ch. 5]{cussler2009diffusion}}  $D$ . { This concentration changes over time due to diffusion as described by Fick's
second law, and determines the spatial and temporal variation of non-uniform
distributions of particles \cite[Ch.2]{howard1993radom630}.}
   
   %This propagation can be mathematically formulated using  Fick's second law \cite{howard1993radom630} 
%\begin{align}
%\frac{{\partial \left( {r \cdot C\left( {r,\left. t \right|{r_0}} \right)} \right)}}{{\partial t}} = D\frac{{{\partial ^2}\left( {r \cdot C\left( {r,\left. t \right|{r_0}} \right)} \right)}}{{\partial {r^2}}}, \label{ficklaw}
%\end{align}
% where $C\left( {r,\left. {t } \right|{r_0}} \right)$  is the molecule concentration  at time ${t }$ and distance $r$ with initial distance $r_0$ between a point transmitter and the center of the receiver, and  $D$ is the  diffusion coefficient, which is usually obtained via experiment  \cite[Ch. 5]{cussler2009diffusion}.

 To reduce the ISI, we introduce  molecule degradation that can occur at any time  via a chemical reaction mechanism in the form of \cite{noel2014improving,heren2015degradation,arman2016reactive}
 \begin{align}
 S\mathop  \to \limits^{{k_d}} P,
 \end{align}  
{ where
 $k_d$ is the degradation rate in $\text{s}^{-1}$, and $P$ is another type of molecule that cannot be recognized by either type of receiver. The degradation rate $k_d$ relates to the
  half-life (${{\Lambda _{{1 \mathord{\left/
 {\vphantom {1 2}} \right.
 \kern-\nulldelimiterspace} 2}}}}$) of messenger molecules via $k_d  = \frac{{\ln 2}}{{{\Lambda _{{1 \mathord{\left/
 {\vphantom {1 2}} \right.
 \kern-\nulldelimiterspace} 2}}}}}$, and $k_d =0$ corresponds to the no  degradation case.  }
 
% Incorporating  molecule degradation into  diffusion, we have 
% \begin{align}
%\frac{{\partial \left( {  C\left( {r,\left. t \right|{r_0}} \right)} \right)}}{{\partial t}} =  - {k_d}C\left( {r,\left. {t } \right|{r_0}} \right),
% \end{align}
%where  $k_d$ relates to the
%  half-life (${{\Lambda _{{1 \mathord{\left/
% {\vphantom {1 2}} \right.
% \kern-\nulldelimiterspace} 2}}}}$) of messenger molecules via $k_d  = \frac{{\ln 2}}{{{\Lambda _{{1 \mathord{\left/
% {\vphantom {1 2}} \right.
% \kern-\nulldelimiterspace} 2}}}}}$, and $k_d =0$ corresponds to the no  degradation case. 
 
\subsection{Reception}
\subsubsection{Absorbing Receiver}
Any information molecules that hit                                                                                                                                                                                                                                                                                                                                                                                                                              the absorbing receiver will be captured for information demodulation. 
This reception process
at the fully absorbing receiver can be  described as{ \cite[Eq. (3.64)]{schulten2000lectures}}
\begin{align}
{\left. {D\frac{{\partial \left( {C^{\text{FA}}\left( {r,\left. t \right|{r_0}} \right)} \right)}}{{\partial r}}} \right|_{r = r_r^ + }} = {k}C^{\text{FA}}\left( {{r_r},\left. t \right|{r_0}} \right),  k \to \infty \label{boundary1}
\end{align}	
where $k$ is the absorption rate (in length$\times $time$^{-1}$).

\subsubsection{Passive receiver}
{ With a transparent membrane at the passive receiver, the information molecules can bypass the surface of the passive receiver freely,  and  molecules within the receiver can be counted at any time  \cite{howard1993radom630}.}
%
%
%\begin{align}
%\frac{{\partial {C}\left( {\left. r_r, t \right|{r_0}} \right)}}{{\partial t}} = {\left. {D\frac{{\partial \left( {C\left( {r,\left. t \right|{r_0}} \right)} \right)}}{{\partial r}}} \right|}, \label{boundary2}
%\end{align}
%which shows that the change  in the molecule concentration over time is equal to the flux of diffusion molecules towards the surface. 
% 

\subsection{Demodulation}\label{dem}
For equivalent comparison, the number  of molecules absorbed by the surface of the absorbing receiver  and the number of observed molecules inside the passive receiver at the end of each bit interval are collected  for information demodulation. 
%For equivalent comparison, we also take samples at the end of each bit interval at both receivers for information demodulation. 
More details of the demodulation at each type of receiver are described as follows.

\subsubsection{Demodulation  criterion at the absorbing receiver}
{ With spherical symmetry, we only need to focus on the number of molecules absorbed  by the surface of the receiver $r = r_r$.}
We consider an absorbing receiver that is capable of counting the net number of molecules absorbed by the surface of the receiver {as} in \cite{Yansha16} by subtracting  the number of
absorbed molecules at the end of the previous bit interval from
that at the end of the current bit interval. The net number of  molecules absorbed over the $j$th bit interval ${N_{\rm{net}}^{\rm{FA}}\left[ j \right] }$ is demodulated as the received signal of the  $j$th  bit (${N^{\rm{Rx}}\left[ j \right] } = {N_{\rm{net}}^{\rm{FA}}\left[ j \right] }$). This is because for the \emph{single} bit transmission at $t=0$, as   time  increases, the number of absorbed molecules increases, which results in increasing ISI, whereas the average net number of absorbed molecules in a given bit interval $T_b$ becomes a constant value as  $T_b$ goes to infinity in a large-scale molecular communication system  as shown in Section IV.   

%We first consider the fixed threshold-based demodulation scheme with the same demodulation threshold for all bits, where   the  net number of absorbed molecules  will accumulate with increasing the number of transmitted bits, and  inevitably impair the system reliability as ISI. To remove these accumulation, we then consider the new demodulation scheme  using a decision feedback detector (DFD) \cite{noel2014overcome} based on the subtraction  between the net number of absorbed molecules in the current bit and that in the  previous bit. \\ 
\subsubsection{Demodulation  criterion at the passive receiver}
With a transparent membrane, the passive receiver is assumed to be capable of counting the number of molecules currently inside the passive receiver at the end of the $j$th bit  interval ${N_{\rm{cur}}^{\rm{PA}}\left[ j \right] }$ for information demodulation (${N^{\rm{Rx}}\left[ j \right] } = {N_{\rm{cur}}^{\rm{PA}}\left[ j \right] }$).   This is because  the  \emph{current} number of observed molecules inside the receiver can remain at a comparable value for a long  time in the large-scale molecular communication system as will be shown in Fig. \ref{Fig2} in Section VI.
%, which results in cumulative ISI with the transmission of multiple bits, and reveals potential benefits for using a DFD.  
For this reason, we only use a simple detector design with one sample collected at the end of each bit interval rather than multiple samples in each bit interval.

{\subsubsection{Demodulation schemes at both receivers}
 We first consider  a fixed threshold-based demodulation with the same decision threshold   $N_{\rm{th}}$ for all bits at both types of receivers, where the receiver demodulates the received signal as bit-1 if ${N^{\rm{Rx}}\left[ j \right] }  \ge N_{\rm{th}}$, and  demodulates the received signal as bit-0 if ${N^{\rm{Rx}}\left[ j \right] } < N_{\rm{th}}$.  In the  fixed threshold-based demodulation, the  received molecules ${N^{\rm{Rx}}\left[ j \right] }$ will accumulate as more bits are transmitted and molecules arrive from more distant transmitters, and  inevitably impair the system reliability as ISI.} 
 
{ To remove this accumulation, we then consider the  demodulation scheme  using a DFD \cite{noel2014overcome} with  the  decision threshold  $N_{\rm{th}}$ at both types of receivers in Section VI, based on the subtraction  between ${N^{\rm{Rx}}}$ in the current bit and that in the  previous bit. More specifically, the receiver demodulates the received signal as bit-1 if $\{{N^{\rm{Rx}}\left[ j \right] } - {N^{\rm{Rx}}\left[ j-1 \right] }\} \ge N_{\rm{th}}$, and  demodulates the received signal as bit-0 if $\{{N^{\rm{Rx}}\left[ j \right] } - {N^{\rm{Rx}}\left[ j-1 \right] }\} < N_{\rm{th}}$.}

  \section{Channel Impulse Response}
  
%The fully absorbing receiver is capable of absorbing any messenger molecule hitting its surface, and it can count the  number of absorbed molecule in any arbitrary time interval for information decoding. In contrast, the passive receiver is transparent to  information molecule, and  is able to count and record the number of free molecule that are within the receiver volume in any time instance for information decoding, without interfering with their diffusion.

  %As shown in Fig.~\ref{System}, 

%
%Thus, the messenger molecules transmitted by other active point transmitters  act as  interference, which impairs the correct reception at the  receiver. To measure this impairment,  we formulate the desired signal, the interfering signal, and the signal-to-interference ratio for  the absorbing receiver and  the passive receiver in the following subsections.

In this section, we present the channel impulse responses at the absorbing receiver and at the passive receiver in the large-scale molecular communication system due to the single bit-1 transmission at each point transmitter.

\subsection{Absorbing Receiver}
\subsubsection{Point-to-point system}
We first provide the background for the receiver observation of  a single point transmitter located distance $r_0$ away from the center  of the absorbing receiver.
To do so, we calculate the rate of absorption at the surface of the absorbing receiver due to the transmitter at distance $r_0$  via \cite[Eq.~(3.106)]{schulten2000lectures}
\begin{align}
K\left( {\left. t \right|{r_0}} \right) = 4\pi r_r^2D{\left. {\frac{{\partial C^{\text{FA}}\left( {\left. {r,t} \right|{r_0}} \right)}}{{\partial r}}} \right|_{r = {r_r}}},
\label{der21} 
\end{align} 
where the molecule distribution function
\begin{align}
C^{\text{FA}}\left( {\left. {r,t} \right|{r_0}} \right) = \frac{1}{{4\pi r{r_0}}}\frac{1}{{\sqrt {4\pi Dt} }}\left( {{e^{ - \frac{{{{\left( {r - {r_0}} \right)}^2}}}{{4Dt}}}} - {e^{ - \frac{{{{\left( {r + {r_0} - 2{r_r}} \right)}^2}}}{{4Dt}}}}} \right),
\label{der22}
\end{align} 
is derived   in \cite{schulten2000lectures}. %{Substituting  \eqref{boundary2} into \eqref{der21}, we have the rate of
%absorption goes to infinity as $k$ goes to infinity. }

Substituting \eqref{der22} into \eqref{der21}, the first hitting probability is derived as 
\begin{align}
K\left( {\left. t \right|{r_0}} \right) = \frac{{{r_r}}}{{{r_0}}}\frac{1}{{\sqrt {4\pi D{t}} }}\frac{{{r_0} - {r_r}}}{t}{e^{ - \frac{{{{\left( {{r_0} - {r_r}} \right)}^2}}}{{4Dt}}}}.
\label{rate}
\end{align}

Taking into account   molecule degradation, the fraction of molecules absorbed by the receiver due to a transmitter at distance $r_0$ during any sampling interval  $[t, t+ T_{ss}]$ with a single impulse pulse  occurring at $t=0$ reduces to  
\begin{align}
F^{{\rm{FA}}} &\left( {\left. {{\Omega _{{r_r}}},t,t + {T_{ss}}} \right|{r_0}} \right) \nonumber \\ & =  F^{{\rm{FA}}}\left( {\left. {{\Omega _{{r_r}}},0,t + {T_{ss}}} \right|{r_0}} \right) - F^{{\rm{FA}}}\left( {\left. {{\Omega _{{r_r}}},0,t} \right|{r_0}} \right),
\label{net_deg}
\end{align}
where
\begin{align}
{{F^{\rm{FA}}}\left( {\left. {{\Omega _{{r_r}}},0, t} \right|{r_0}} \right)}&  = \int_0^{t} {K\left( {\left. t \right|{r_0}} \right){e^{ - {k_d} t}}dt}
\nonumber \\& \hspace{-2.5cm} = \frac{{{r_r}}}{{{r_0}}}\exp \left( { - \sqrt {\frac{{k_d} }{D}} \left( {{r_0} - {r_r}} \right)} \right) - \frac{{{r_r}}}{{2{r_0}}}{\exp\left( - \sqrt {\frac{{k_d} }{D}} \left( {{r_0} - {r_r}} \right)\right)}
\nonumber \\& \hspace{-2.0cm} \left[ {\erf\left( {\frac{{{r_0} - {r_r}}}{{\sqrt {4Dt} }} - \sqrt {k_d}  t} \right) + {\exp\left({2\sqrt {\frac{{k_d}}{D}} \left( {{r_0} - {r_r}} \right)}\right)}} \right.
\nonumber \\& \hspace{-2.0cm} \left. {\left( {\erf\left( {\frac{{{r_0} - {r_r}}}{{\sqrt {4Dt} }} + \sqrt {k_d}  t} \right) - 1} \right) + 1} \right]
\label{der23_deg}.
\end{align}
We note that \eqref{der23_deg} is derived following the method {for  the point-to-point system} in \cite[Eq. (12)]{heren2015degradation}. We  see that increasing $k_d$ decreases the fraction of  molecules  absorbed by the absorbing receiver.

 Without molecule degradation ($k_d = 0$), ${{F^{\rm{FA}}}\left( {\left. {{\Omega _{{r_r}}},0, t} \right|{r_0}} \right)}$ simplifies to  \cite[Eq. (32)]{Yilmaz14}
\begin{align}
{{F^{\rm{FA}}}\left( {\left. {{\Omega _{{r_r}}},0, t} \right|{r_0}} \right)}&  =   \frac{{{r_r}}}{{r_0}}\erfc\Big\{ {\frac{{r_0-{r_r}}}{{\sqrt {4D{t}} }}}  \Big\}.
\label{der23_nodeg}
\end{align}
%The net number of  absorbed molecules  at the receiver at distance $r_0$ from a transmitter in the $j$th bit with a single pulse  occurring at $t=0$ under molecule degradation  is expressed as \cite{Yansha16}
%\begin{align}
%{N_{\rm{net}}^{\rm{FA}} |{r_0} \left[ j \right] } \sim B\left( {N_{{\rm{tx}}}^{{\rm{FA}}},{F^{{\rm{FA}}}_{\rm{Deg}}}\left( {\left. {{\Omega _{{r_r}}},\left( {j - 1} \right){T_b},j{T_b}} \right|{r_0}} \right)} \right),
%\label{der24}
%\end{align}
%where ${F^{{\rm{FA}}}_{\rm{Deg}}}\left( {\left. {{\Omega _{{r_r}}},\left( {j - 1} \right){T_b},j{T_b}} \right|{r_0}} \right)$ is given in \eqref{der23_deg}.
%In a point-to-point molecular communication system with a single point source located distance $d$ away from the closest point on the surface of a single absorbing receiver, the fraction of molecule absorbed at the receiver until time $t$, due to  the molecule emission occurring at $t=0$, is derived as \cite{Yilmaz14}
%\begin{align}
%F\left( {\left. {{\Omega _{{r_r}}},t} \right|d} \right) = \frac{{{r_r}}}{{d + {r_r}}}\erfc\Big\{ {\frac{d}{{\sqrt {4Dt} }}} \Big\}, \label{Fraction_d}
%\end{align}
%where $D$ is the constant diffusion coefficient, which is usually obtained via experiment as in \cite[Ch. 5]{cussler2009diffusion}.

\subsubsection{Large-scale  system}

In our proposed large-scale  system, the center of an absorbing receiver is fixed at the origin of  a  3D fluid environment. 

 Using  the Slivnyak-Mecke's theorem \cite{baccelli2009stochastic}, the  fraction  of absorbed molecules  at the   receiver during any sampling interval  $[t, t+T_{ss}]$ due to an \emph{arbitrary} point transmitter $x$ at the location $\textbf{x}$ emitting a single pulse at $t=0$  ${{F^{\rm{FA}}}\left( {\left. {{\Omega _{{r_r}}},t, t+T_{ss}} \right|{\left\| \textbf{x} \right\|}} \right)}$ can be obtained via \eqref{der23_deg},
 %\cite{Yilmaz14}
%\begin{align}
%{{F^{\rm{FA}}}\left( {\left. {{\Omega _{{r_r}}},t, t+T_b} \right|{\left\| \textbf{x} \right\|}} \right)}&   =
% {{F^{\rm{FA}}_{{\rm{Nde}}}}\left( {\left. {{\Omega _{{r_r}}},t, t+T_b} \right|{\left\| \textbf{x} \right\|}} \right)},
%\label{der23}
%\end{align}
%for no molecule degradation, and 
%\begin{align}
%{{F^{\rm{FA}}}\left( {\left. {{\Omega _{{r_r}}},t, t+T_{ss}} \right|{\left\| \textbf{x} \right\|}} \right)}&   =
% {{F^{\rm{FA}}_{{\rm{Deg}}}}\left( {\left. {{\Omega _{{r_r}}},t, t+T_{ss}} \right|{\left\| \textbf{x} \right\|}} \right)},
%\label{der231}
%\end{align}
 where $\left\| \textbf{x} \right\|$ is the distance between the point transmitter and the center of the  receiver where the transmitters follow a HPPP.
 %, ${{F^{\rm{FA}}_{{\rm{Deg}}}}\left( {\left. {{\Omega _{{r_r}}},t, t+T_{ss}} \right|{\left\| \textbf{x} \right\|}} \right)}$ is given in  \eqref{der23_deg}.
 
Recalling that the propagation of each molecule is independent,  
 the cumulative fraction $F^{\rm{FA}}_{\rm{all}}$ of  absorbed  molecules at the   receiver during any sampling interval  $[t, t+T_{ss}]$ due to all  active point transmitters   emitting a single pulse at $t=0$ can be formulated as
 \begin{align}
& {{F_{\rm{all}}^{{\rm{FA}}}}\left( { {{\Omega _{{r_r}}},t, t+T_{ss}} } \right)} =\sum\limits_{x \in{\Phi _a}} {{F^{\rm{FA}}}\left( {\left. {{\Omega _{{r_r}}},t, t+T_{ss}} \right|{\left\| \textbf{x} \right\|}} \right)}  ,
\label{total_FA}
\end{align}
where  ${{F^{\rm{FA}}}\left( {\left. {{\Omega _{{r_r}}},t, t+T_{ss}} \right|{\left\| \textbf{x} \right\|}} \right)}$ can be obtained from \eqref{der23_deg}.

The  expected  \emph{net} number of   molecules  absorbed by the receiver during any sampling interval   $[t, t+T_{ss}]$ due to all of the active point transmitters   emitting a single pulse  at $t=0$   can be   calculated as 
\begin{align}
{\mathbb{E}}\left\{ {N_{{\rm{all}}}^{{\rm{FA}}}\left( { {{\Omega _{{r_r}}},t, t+T_{ss}} } \right)} \right\}  
& = {N_{\rm{tx}}} {\mathbb{E}}\Big\{{F_{\rm{all}}^{{\rm{FA}}}\left( { {{\Omega _{{r_r}}},t, t+T_{ss}} } \right)}\Big\}  
,  \label{FA_uo}
\end{align}
where $N_{{\rm{all}}}^{{\rm{FA}}}\left( { {{\Omega _{{r_r}}},t, t+T_{ss}} } \right)$ is the net number of absorbed molecules, and 
 $ {F_{\rm{all}}^{{\rm{FA}}}\left( { {{\Omega _{{r_r}}},t, t+T_{ss}} } \right)}$ is given in \eqref{total_FA}.

It is well known that the distance between the transmitter and the receiver  in molecular communication is the main contributor to the    signal strength.  In the absence of flow,
 the  nearest transmitter  will provide the strongest signal for the receiver.  In order to examine the impact of the signal from the nearest transmitter on the received signal in the large-scale molecular communication system, we present
the expected number  of   absorbed  molecules at this  receiver during any sampling interval  $[t, t+T_{ss}]$ due to a single  pulse emission by the \emph{nearest}  transmitter  as
\begin{align}
{{\mathbb{E}}_{\rm{u}}^{{\rm{FA}}}} = {\mathbb{E}}\Big\{{{F^{\rm{FA}}}\left( {\left. {{\Omega _{{r_r}}},t,t+T_{ss}} \right|{\left\| \textbf{x}^* \right\|}} \right)}\Big\} , \label{useful}
\end{align}
where $\left\| \textbf{x}^* \right\|$ denotes the distance between   the receiver and the nearest  transmitter,
\begin{align}\label{cell_selection_II}
 \quad x^{*} = \mathop {\arg \min }\limits_{x \in {\Phi _a}}  \left\| {{\textbf{x}}} \right\|, 
\end{align}
 $x^{*}$  denotes  the nearest point  transmitter  for the receiver, and ${\Phi _a}$ denotes the  set of active transmitters' positions.

To examine the impact of the aggregate signal from the remaining transmitters, we present   the expected number  of   absorbed  molecules at this  receiver during any sampling interval  $[t, t+T_{ss}]$ due to  single  pulse emissions by the other    transmitters    as
\begin{align}
 {{\mathbb{E}}_{\rm{o}}^{{\rm{FA}}}}  = {\mathbb{E}}\Bigg\{\sum\limits_{{{x \in{\Phi _a}} \mathord{\left/
 {\vphantom {{{\Phi _a}} {{{x}^*}}}} \right.
 \kern-\nulldelimiterspace} {{{x}^*}}}} {{F^{\rm{FA}}}\left( {\left. {{\Omega _{{r_r}}},t, t+T_{ss}} \right|{\left\| \textbf{x} \right\|}} \right)}\Bigg\}, \label{interference}
\end{align}
where  ${{F^{\rm{FA}}}\left( {\left. {{\Omega _{{r_r}}},t, t+T_{ss}} \right|{\left\| \textbf{x} \right\|}} \right)}$ is given in  \eqref{der23_deg}.

\subsection{Passive Receiver}
\subsubsection{Point-to-point system}
{In a point-to-point molecular communication system with a single point transmitter located at ${\overrightarrow{r}}$ relative to the center of a passive receiver with radius $r_r$, the local point concentration at the center of the passive receiver at time $t$ due to a single  pulse emission by the transmitter occurring at $t=0$ is given as \cite[Eq. (4.28)]{nelson2004biological}
\begin{align}
C\left( {\left. {{ {{\Omega}} _{{r_r}} },t} \right|{\overrightarrow{r}}} \right) = \frac{1}{{{{\left( {4\pi Dt} \right)}^{{3 \mathord{\left/
 {\vphantom {3 2}} \right.
 \kern-\nulldelimiterspace} 2}}}}}\exp \Big( { - \frac{{{{ \left| {\overrightarrow{r}} \right|
 }^2}}}{{4Dt}}} \Big), \label{Fraction_d_pass}
\end{align}
where $\overrightarrow{r}  = [x,y,z]$,  and $[x,y,z]$ are  the coordinates along the three axes.}

The fraction of   molecules observed  inside the passive receiver with volume  ${V_{{\Omega _{{r_r}}}}}$ at time $t$  is denoted as
{\begin{align}
{F^{\rm{PS}}}\left( {\left. {{\Omega _{{r_r}}},t} \right|{\left| {\overrightarrow r } \right|}} \right) & = \int\limits_{{{{\Omega _{{r_r}}}}}} {C\left( {\left. {{{\Omega}_{{r_r}}},t} \right|{\overrightarrow{r}}} \right) \diff{{{{\Omega _{{r_r}}}}}}} . \label{PS_e}
\end{align}}

In most molecular communication literature considering a passive receiver, the uniform concentration assumption inside the passive receiver is  applied, which immediately results in the  fraction of observed molecules inside the passive receiver as 
{\begin{align}
{F^{\rm{PS}}}\left( {\left. {{\Omega _{{r_r}}},t} \right|{\left| {\overrightarrow r } \right|)}} \right) & \approx {C\left( {\left. {{{\Omega}_{{r_r}}},t} \right|\left| {\overrightarrow r } \right|} \right) {{V_{{\Omega _{{r_r}}}}}}} , \label{PS_uniform}
\end{align}}
however, this result
 relies on the receiver being sufficiently far from the transmitter (see  \cite{noel2013using}), which we \emph{cannot} guarantee here since the transmitters are placed randomly.
 
 With the actual \emph{non-uniform} concentration inside the passive receiver, the  fraction of observed molecules inside the passive receiver is calculated as 
 {\begin{align}
{F^{{\rm{PS}}}}\left( {\left. {{\Omega _{{r_r}}},t} \right|\left| {\overrightarrow r } \right|} \right){\rm{ }} = \int\limits_0^{{r_r}} {\int\limits_0^{2\pi } {\int\limits_0^\pi  {C\left( {\left. {{{\Omega}_{{r_r}}},t} \right| {\overrightarrow r } } \right)r^2 \sin \theta } } } \diff \theta \diff \phi \diff r  . \label{PS_nonuniform}
\end{align}}

{The molecule degradation introduces a decaying exponential term as in \cite[Eq. (10)]{noel2014improving}. Therefore, according to \eqref{PS_nonuniform} and Theorem 2 in \cite{noel2013using},   the fraction ${F^{\rm{PS}}}$ of   molecules observed  inside the passive receiver at time $t$ due to a single  pulse emission by a transmitter at $r_0$ away from  the center of a passive receiver with radius $r_r$ at time $t=0$  is derived  as}
\begin{align}
&{F^{\rm{PS}}}\left( {\left. {{\Omega _{{r_r}}},t} \right|{r_0}} \right)  =  {e^{ - {k_d} t}} \Bigg[\frac{1}{2}\Big[ {{\rm{erf}}\Big( {\frac{{{r_r} - {r_0}}}{{2\sqrt {Dt} }}} \Big) + {\rm{erf}}\Big( {\frac{{{r_r} + {r_0}}}{{2\sqrt {Dt} }}} \Big)} \Big] \Bigg.
\nonumber \\ & \Bigg. +{\frac{{\sqrt{Dt}}}{\sqrt{\pi }{{r_0}}}} \Big[ {\exp \Big( { - \frac{{{{\left( {{r_r} + {r_0}} \right)}^2}}}{{4Dt}}} \Big) - \exp \Big( { - \frac{{{{\left( {{r_0} - {r_r}} \right)}^2}}}{{4Dt}}} \Big)} \Big]\Bigg]. \label{PS_nonuniform_der_r0}
\end{align}
We see from  \eqref{PS_nonuniform_der_r0} that increasing $k_d$ decreases the fraction of molecules observed at the passive receiver.
\subsubsection{Large-scale   system}
In the large-scale   molecular communication system with a passive receiver centered at the  origin,
 the fraction $F^{\rm{PS}}$ of  molecules  observed inside  the  passive receiver at time   $T_{ss}$ due to an \emph{arbitrary} point transmitter $x$ at the location $\textbf{x}$ emitting a single pulse at $t=0$, $ {{F^{\rm{PS}}}\left( {\left. {{\Omega _{{r_r}}},t} \right|{\left\| \textbf{x} \right\|}} \right)}$ can be obtained using \eqref{PS_nonuniform_der_r0}.

 Due to the independent propagation of each molecule,  the expected number of  molecules observed inside the receiver at time $T_{ss}$ due to a single  pulse emission by \emph{all}  transmitters at $t=0$ is given as 
\begin{align}
{\mathbb{E}}\left\{ {N_{\rm{all}}^{{\rm{PS}}}\left( { {{\Omega _{{r_r}}},T_{ss}} } \right)} \right\} = &\underbrace {N_{{\rm{tx}}}{\mathbb{E}}\Big\{{F^{{\rm{PS}}}}\left( {\left. {{\Omega _{{r_r}}},T_{ss}} \right|\left\| {{{\bf{x}}^*}} \right\|} \right)\Big\}}_{{\mathbb{E}}_{\rm{u}}^{{\rm{PS}}}}
\nonumber\\ & \hspace{-1.3cm} + \underbrace {{N_{\rm{tx}}}{\mathbb{E}}\Bigg\{{\sum _{x \in {{{\Phi _a}} \mathord{\left/
 {\vphantom {{{\Phi _a}} {{x^*}}}} \right.
 \kern-\nulldelimiterspace} {{x^*}}}}}{F^{{\rm{PS}}}}\left( {\left. {{\Omega _{{r_r}}},T_{ss}} \right|\left\| {\bf{x}} \right\|} \right)\Bigg\}}_{{\mathbb{E}}_{\rm{o}}^{{\rm{PS}}}}
, \label{PS_all}
\end{align}
where ${{F^{\rm{PS}}}\left( {\left. {{\Omega _{{r_r}}},T_{ss}} \right|{\left\| \textbf{x} \right\|}} \right)}$ is obtained using \eqref{PS_nonuniform_der_r0}, ${{\mathbb{E}}_{\rm{u}}^{{\rm{PS}}}}$ is the expected number of  molecules observed inside the receiver at time $T_{ss}$ due to the nearest transmitter, and ${{\mathbb{E}}_{\rm{o}}^{{\rm{PS}}}}$ is the expected number of  molecules observed inside the receiver at time $T_{ss}$ due to the other  transmitters. 

%{
%The average SIR based on the expected number  of  molecules observed inside the passive receiver at time $T_b$  is defined as 
%\begin{align}
%\rm{SIR{_{PS}}} = &\frac{{N_{\rm{tx}}^{\rm{PS}}}{\mathbb{E}}\left\{{F^{\rm{PS}}}\left( {\left. {{\Omega _{{r_r}}},T_b} \right|{\left\| {\textbf{x}}^* \right\|}} \right)\right\}}{{N_{\rm{tx}}^{\rm{PS}}}{\mathbb{E}}\Big\{\sum\limits_{{{{\Phi _a}} \mathord{\left/
% {\vphantom {{{\Phi _a}} {{{\textbf{x}}^*}}}} \right.
% \kern-\nulldelimiterspace} {{{\bf{x}}^*}}}}{F^{\rm{PS}}}\left( {\left. {{\Omega _{{r_r}}},T_b} \right|{\left\| \textbf{x} \right\|}} \right)\Big\}}= \frac{{{{\mathbb{E}}_{\rm{u}}^{{\rm{PS}}}}}}{{{{\mathbb{E}}_{\rm{I}}^{{\rm{PS}}}}}}, \label{SIRPS_raw}
%\end{align}
%
%}
%where ${{\mathbb{E}}_{\rm{u}}^{{\rm{PS}}}}$ is the expected number of  molecules observed inside the receiver at time $T$ due to the nearest transmitter, and ${{\mathbb{E}}_{\rm{I}}^{{\rm{PS}}}}$ is expected number of  molecules observed inside the receiver at time $T$ due to the other (interfering) transmitters.}

\section{Receiver Observations}

In this section, we first derive  the distance distribution between the receiver and the nearest point transmitter. Throughout this section, we focus on the receiver observations at the receivers due to a single emission at each point transmitter at $t=0$. To understand the impact of individual TXs relative to the aggregate signal,   we 
derive  exact expressions for the expected number of  molecules observed at  the receiver  due to the nearest point  transmitter and that due to the other  transmitters. We then present exact expressions for the expected number of molecules observed at the receiver due to all transmitters.
%We also present  exact expressions for the SIRs. 

\subsection{Distance Distribution}
Unlike the  stochastic geometry modelling of wireless networks, where the transmitters are randomly located in  unbounded space, we impose that the point transmitters in a molecular communication system can only be distributed \emph{outside} the surface of the spherical receiver. Taking into account the minimum distance $r_r$ between point transmitters and the   receiver center, we derive the  PDF of the shortest distance between a point transmitter and the receiver with radius $r_r$ in the following proposition.

\begin{ppro}
The PDF of the shortest distance between any point transmitter and the receiver with radius $r_r$  in 3D space  is 
\begin{align}
{f_{\left\| {{\bf{x}^*}} \right\|}}(x) = 4\lambda_a\pi {x^2}{e^{ - \lambda_a \left( {\frac{4}{3}\pi {x^3} - \frac{4}{3}\pi {r_r}^3} \right)}}. \label{shortest_dist}
\end{align}
%where $\lambda_a =\lambda \rho_{a}$.
\end{ppro}
\begin{proof}
See Appendix A.
\end{proof}

Based on the proof of Proposition 1, we also derive the PDF of the shortest distance between any point transmitter and the receiver  in 2D space in the following lemma.
\begin{cor}
The PDF of the shortest distance between any  point transmitter and the receiver  in 2D space  is given by 
{
\begin{align}
{f_{\left\| {{x^*}} \right\|}}(x) = 2\lambda_a \pi x{e^{ - \lambda_a \left( {\pi {x^2} - \pi {r_r}^2} \right)}}, \label{shortest_dist_2D}
\end{align}}
where $\lambda_a =\lambda \rho_{a}$.
\end{cor}
With $r_r =0$, Corollary 1 reduces to \cite[Eq. (19)]{novlan2013uplink}.

\subsection{General Expected  Receiver Observations}
In this subsection,  we first derive  simple expressions for the expected number of  molecules observed at the  receiver due to the nearest  transmitter and the other  transmitters  to demonstrate their relative impact on the expected receiver observations. 

Using  Campbell’s theorem \cite[Eq. (1.18)]{baccelli2009stochastic} and Proposition 1, the   expected  \emph{net} number of   molecules observed  during any sampling  interval $[t,t+T_{ss}]$ at   the   receiver  due to the nearest transmitter and the other transmitters are derived as
\begin{align}
{{{\mathbb{E}}_{\rm{u}}}}= \;& 4{\lambda _a}\pi N_{{\rm{tx}}}{e^{\frac{4}{3}\pi {r_r}^3{\lambda _a}}}\int_{{r_r}}^\infty  {\Phi \left( r \right){r^2}\exp \Big\{ { - \frac{4}{3}\pi {r^3}{\lambda _a}} \Big\}\diff r} 
, \label{E_u_FA_PS}
\end{align}
and
\begin{align}
{{{\mathbb{E}}_{\rm{o}}}} = \;&  {\left( {4\pi {\lambda _a}} \right)^2}{e^{\frac{4}{3}\pi {r_r}^3{\lambda _a}}}N_{{\rm{tx}}}\int_{{r_r}}^\infty  {\int_x^\infty  {\Phi \left( r \right)} }{r^2}\diff r 
 \nonumber\\& \times {x^2}{e^{ - \frac{4}{3}\pi {x^3}{\lambda _a}}}\diff x, \label{E_I_PS}
\end{align}
respectively, where
\begin{align}
\Phi \left( r \right) = {F^{{\rm{FA}}}}\left( {\left. {{\Omega _{{r_r}}},t,t + {T_{ss}}} \right|r} \right), \label{phi_FA}
\end{align}
 for  the absorbing receiver, and
\begin{align}
\Phi \left( r \right) = {F^{{\rm{PS}}}}\left( {\left. {{\Omega _{{r_r}}},t + {T_{ss}}} \right|r} \right) - {F^{{\rm{PS}}}}\left( {\left. {{\Omega _{{r_r}}},t} \right|r} \right), \label{phi_PS}
\end{align}
for  the passive receiver. In \eqref{phi_FA}, $ {F_{{\rm{FA}}}\left( { {{\Omega _{{r_r}}},t, t+T_{ss}} } \right)}$ is the fraction of  molecules  absorbed by the absorbing receiver given in \eqref{net_deg}. 
In \eqref{phi_PS}, $ {F_{{\rm{PS}}}\left( { {{\Omega _{{r_r}}},t} } \right)}$ is the fraction  of   molecules observed  inside the passive receiver  given in \eqref{PS_nonuniform_der_r0}. 
 We  observe that ${{{\mathbb{E}}_{\rm{u}}}}$ and ${{{\mathbb{E}}_{\rm{o}}}}$ both  increase proportionally with the density of  transmitters. 
%{ and an analytical expression for the average SIR  at the absorbing receiver in 3D space.}
 
% Using  Campbell’s theorem \cite[Eq. (1.18)]{baccelli2009stochastic} and Proposition 1, the   expected   number of   molecules being  absorbed  during any bit interval $[t,t+T_{b}]$   at   the   receiver  due to the nearest transmitter  is derived as
%\begin{align}
%{{{\mathbb{E}}_{\rm{u}}^{{\rm{FA}}}}}= &{N_{\rm{tx}}^{\rm{FA}}}\int_{{r_r}}^\infty  {\Pi_0^t \left( x \right)4{\lambda _a}\pi {x^2}}  
%{e^{ - {\lambda_a} \left( {\frac{4}{3}\pi {x^3} - \frac{4}{3}\pi {r_r}^3} \right)}}dx
%, \label{E_u_FA}
%\end{align}
%and
%\begin{align}
%&{{{\mathbb{E}}_{\rm{o}}^{{\rm{FA}}}}} = \nonumber\\&\hspace{0.2cm} {N_{\rm{tx}}^{\rm{FA}}}\left(4\pi {\lambda _a} \right)^2{  \int_{{r_r}}^\infty  {\int_x^\infty  \Pi_0^t \left( r \right)} } {r^2}dr
% {x^2}{e^{ - {\lambda _a}\left( {\frac{4}{3}\pi {x^3} - \frac{4}{3}\pi {r_r}^3} \right)}}dx, \label{E_I_FA}
%\end{align}
%respectively, where $\Pi_0 ^t\left( r \right)$ is given in \eqref{phi}. 

%\subsection{General Expected  Receiver Observations}

We now derive   the expected  net number of   molecules   observed at the  receiver in the following theorem.
\begin{theorem}
%The  expected net number of  molecules observed inside the passive receiver during any sampling time interval $[t,t+T_{ss}]$ in 3D space  is derived as
%\begin{align}
%&{\mathbb{E}}\left\{ {N_{{\rm{all}}}^{{\rm{PS}}}\left( {\left. {{\Omega _{{r_r}}},t,t + {T_{ss}}} \right|\left\| {\bf{x}} \right\|} \right)} \right\}= 4N_{\rm{tx}}^{\rm{PS}}\pi {\lambda _a}
%\nonumber\\ &\left[ {\int_{{r_r}}^\infty  {{F^{{\rm{PS}}}}\left( {\left. {{\Omega _{{r_r}}},t + {T_{ss}}} \right|r} \right)} {r^2}dr} \right.
%\left. { - \int_{{r_r}}^\infty  {{F^{{\rm{PS}}}}\left( {\left. {{\Omega _{{r_r}}},t} \right|r} \right)} {r^2}dr} \right]
%, \label{Total_PS_final}
%\end{align}

The  expected net  number of  molecules observed at the receiver during any sampling interval $[t,t+T_{ss}]$ due to all transmitters emitting  single pulses at $t=0$ is derived as
\begin{align}
&{\mathbb{E}}\left\{ {N_{{\rm{all}}}\left( {\left. {{\Omega _{{r_r}}},t,t+T_{ss} } \right|\left\| {\bf{x}} \right\|} \right)} \right\}= 4{N_{\rm{tx}}}\pi {\lambda _a}{\int_{{r_r}}^\infty  \Phi \left( r \right) {r^2}\diff r} 
, \label{Total_PS_final_t}
\end{align}
where   $\Phi \left( r \right)$   is given in \eqref{phi_FA} for the absorbing receiver  and \eqref{phi_PS} for the passive receiver.

\end{theorem}
\begin{proof}
See Appendix B.
\end{proof}
From Theorem 1,  we find that the  expected net number of  observed molecules at the  receiver   is linearly proportional to the density of  transmitters, which will positively improve
the peak  observation, but negatively
bring increased ISI.

\subsection{Absorbing Receiver without Molecule Degradation}
To obtain additional insights, we
 now present the exact and asymptotic   expressions for   the  expected net number of  molecules  absorbed  by the absorbing receiver without molecule degradation in closed-form.  We  only consider  the absorbing receiver  here because it leads to a simple insightful expression.
 
 %Unfortunately, the expect net number of molecules observed at the passive receiver cannot be solved in simple closed-form due to the error functions inside integral  in \eqref{Total_PS_final_t} and \eqref{PS_nonuniform_der_r0}.
\begin{lemma}
With $k_d=0$, the   expected  net number of molecules absorbed by the  absorbing receiver in 3D space during any sampling interval $[t,t+T_{ss}]$ is derived as
\begin{align}
{\mathbb{E}}&\left\{ N_{{\rm{all}}}^{{\rm{FA}}}\left( { {{\Omega _{{r_r}}},t,t + {T_{ss}}} } \right) \right\}  
\nonumber\\&= 4{N_{{\rm{tx}}}}\sqrt \pi  {\lambda _a}{r_r}\left[ {D\sqrt \pi  {T_{ss}} + 2\sqrt D {r_r}\left( {\sqrt {{T_{ss}} + t}  - \sqrt t } \right)} \right].  \label{Total_FA_final}
\end{align}

The   expected  total number of   molecules being absorbed by time $t$ at   the  absorbing receiver in 3D space  is derived as
\begin{align}
{\mathbb{E}}&\left\{ N_{{\rm{all}}}^{{\rm{FA}}}\left( { {{\Omega _{{r_r}}},0,t} } \right) \right\}  
= 4{N_{{\rm{tx}}}}\sqrt \pi  {\lambda _a}{r_r}\left[ {D t\sqrt \pi  + 2{r_r}\sqrt {D t} } \right].  \label{Total_FA_final_t}
\end{align}
%By substituting $T=0$ into \eqref{Total_FA_final}, we derive the cumulative total  number of absorbed molecule until time $T$ as
%\begin{align}
%{\mathbb{E}}&\left\{ F_{{\rm{tot}}}^{{\rm{FA}}}\left( {\left. {{\Omega _{{r_r}}},T} \right|\left\| \textbf{x}^* \right\|} \right) \right\}  
%\nonumber\\&= {N_{{\rm{tx}}}^{\rm{FA}}}4\sqrt \pi  {\lambda _a}{r_r}\left[ {D\sqrt \pi  {T_{ss}} + \sqrt D {r_r}2\left( {\sqrt {{T_{ss}}} } \right)} \right].  \label{Total_FA_final}
%\end{align}

\end{lemma}
\begin{proof}
See Appendix C.
 \end{proof}
 
 From Lemma 1,  we find that the  expected net number of  molecules absorbed by the absorbing receiver    increases with increasing  diffusion coefficient or receiver radius. As expected, we find that the expected \emph{total} number of   molecules  absorbed by time $ t$  is always increasing with  $ t$ and does not converge, even though there was only one release by each transmitter.

Next, we  examine the asymptotic results for the  expected  net number of  molecules absorbed by the  absorbing receiver during any sampling interval $[t,t+T_{ss}]$  as $t \to \infty$ to find the maximum expected net number of  absorbed molecules.
\begin{lemma}
With $k_d=0$ and as $t \to \infty$, the expected  net number of  molecules absorbed by the  absorbing receiver during any sampling interval $[t,t+T_{ss}]$ in 3D space  is derived as
\begin{align}
{\mathbb{E}}&\left\{ {N_{{\rm{all}}}^{{\rm{FA}}}\left( {{\Omega _{{r_r}}},t, t+{T_{ss}}} \right)} \right\} \mathop  = \limits^{t \to \infty }  4\pi N_{{\rm{tx}}}{\lambda _a}{r_r}D{T_{ss}}.  \label{Total_FA_final_asy}
\end{align}
\end{lemma}
Lemma 2 reveals that as time sufficiently increases, the expected  net number of  molecules absorbed by the  absorbing receiver becomes a constant determined by the sampling interval. More importantly, this also reveals that the expected  net number of  absorbed molecules during the bit interval 
increases with the number of transmitted symbols (i.e., ISI).

\section{Error Probability}
In this section, we move from the expected receiver observations to the instantaneous receiver observations and the  bit error probability of the large-scale  molecular communication system with the absorbing receiver and the passive receiver under  molecule degradation.
This section focuses on simple detectors requiring one sample per bit, where  the net number of molecules absorbed by the surface of the absorbing  receiver during each bit interval, and the number of molecules observed inside the passive receiver at the end of each bit interval, are sampled for information demodulation.  The bit error probability of the proposed system with a DFD involves the subtraction of two dependent variables as shown in Section \ref{dem}, which is  analytically non-trivial to derive.

\subsection{Instantaneous  Absorbing Receiver Observations}

We first present the net number of    molecules  absorbed by the receiver  in the $j$th bit due to all the  point transmitters ${\Phi _a}$  with multiple transmitted bits   as 
\begin{align}
&N_{{\rm{net}}}^{{\rm{FA}}}\left[ j \right]\sim \sum\limits_{{{x \in{\Phi _a}} }} \sum\limits_{i = 1}^j {b_i}
\nonumber\\& \hspace{1cm} \times B\Big( {N_{{\rm{tx}}},{{F^{\rm{FA}}}\left( {\left. {{\Omega _{{r_r}}},\left( {j - i} \right){T_b},(j-i+1){T_b}} \right|{\left\| \textbf{x} \right\|}} \right)} } \Big),
\label{net_FA}
\end{align}
%where 
%\begin{align}
%&{F_{{\rm{all}}}^{{\rm{FA}}}\left( {{\Omega _{{r_r}}},\left( {j - i} \right){T_b},\left( {j - i + 1} \right){T_b}} \right)} \nonumber \\ & \hspace{0.5cm}  = \sum\limits_{{{x \in{\Phi _a}} }}{{F^{\rm{FA}}}\left( {\left. {{\Omega _{{r_r}}},\left( {j - i} \right){T_b},(j-i+1){T_b}} \right|{\left\| \textbf{x} \right\|}} \right)},
%\label{total_FA_Tb}
%\end{align}
%and
%\begin{align}
%& {{F^{\rm{FA}}}\left( {\left. {{\Omega _{{r_r}}},\left( {j - i} \right){T_b},(j-i+1){T_b}} \right|{\left\| \textbf{x} \right\|}} \right)}  =
%\nonumber \\&  \hspace{-0.2cm} \frac{{{r_r}}}{{\left\| \textbf{x} \right\|}}\erfc\Big\{ {\frac{{\left\| \textbf{x} \right\|-{r_r}}}{{\sqrt {4D{((j-i+1){T_b})}} }}}  \Big\} -  \frac{{{r_r}}}{{\left\| \textbf{x} \right\|}}\erfc\Big\{ {\frac{{\left\| \textbf{x} \right\|-{r_r}}}{{\sqrt {4D{(j-i){T_b}}} }}}  \Big\},
%\label{FA_s}
%\end{align}
where ${{F^{\rm{FA}}}\left( {\left. {{\Omega _{{r_r}}},\left( {j - i} \right){T_b},(j-i+1){T_b}} \right|{\left\| \textbf{x} \right\|}} \right)}$ can be obtained via
\eqref{net_deg}, 
and $b_i$ is the $i$th transmitted bit.

The sum of binomial random variables in \eqref{net_FA} does not lend itself to easy evaluation, thus we apply the Poisson approximation as in \cite{Yansha16} to represent \eqref{net_FA} as 
\begin{align}
N_{{\rm{net}}}^{{\rm{FA}}}\left[ j \right]\sim P &\Big( {N_{{\rm{tx}}}\sum\limits_{i = 1}^j {{b_i}{\sum _{x \in {\Phi _a}}}} } \Big.
\nonumber \\& \Big. {F^{{\rm{FA}}}\left( {\left. {{\Omega _{{r_r}}},\left( {j - i} \right){T_b},(j - i + 1){T_b}} \right|\left\| {\bf{x}} \right\|} \right)} \Big).
\label{net_FA_po}
\end{align}

\subsection{Instantaneous  Passive Receiver Observations}
The number of    molecules  observed inside the passive receiver  in the $j$th bit due to all the active point transmitters $x$  with multiple transmitted bits   is expressed as 
\begin{align}
{N_{\rm{cur}}^{\rm{PS}}  \left[ j \right] } \sim  \sum\limits_{{{x \in{\Phi _a}} }} \sum\limits_{i = 1}^j {{b_i}B\left( {N_{{\rm{tx}}},F^{{\rm{PS}}}\left( {{\Omega _{{r_r}}},\left( {j - i + 1} \right){T_b}} \right)} \right)} ,
\label{cur_PS}
\end{align}
where 
%\begin{align}
%{{F_{\rm{all}}^{{\rm{PS}}}}\left( { {{\Omega _{{r_r}}},\left( {j - i + 1} \right){T_b}} } \right)}&   = \sum\limits_{{{x \in{\Phi _a}} }}{{F^{\rm{PS}}}\left( {\left. {{\Omega _{{r_r}}},\left( {j - i + 1} \right){T_b}} \right|{\left\| \textbf{x} \right\|}} \right)},
%\label{total_PS_j}
%\end{align}
%and
${{F^{\rm{PS}}}\left( {\left. {{\Omega _{{r_r}}},\left( {j - i + 1} \right){T_b}} \right|{\left\| \textbf{x} \right\|}} \right)}$ can be obtained  via \eqref{PS_nonuniform_der_r0}.

Using the Poisson approximation, we  write \eqref{cur_PS} as 
\begin{align}
&N_{{\rm{cur}}}^{{\rm{PS}}}\left[ j \right]\sim
 P\Big( {N_{{\rm{tx}}}\sum\limits_{i = 1}^j {{b_i}\sum\limits_{x \in {\Phi _a}} {{F^{{\rm{PS}}}}\left( {\left. {{\Omega _{{r_r}}},(j - i + 1){T_b}} \right|\left\| {\bf{x}} \right\|} \right)} } } \Big).
\label{cur_PS_po}
\end{align}

\subsection{General Bit Error Probability }
Based on \eqref{net_FA_po} and \eqref{cur_PS_po}, we can unify the demodulation variable at both receivers for simplicity  as 
\begin{align}
&N\left[ j \right]\sim
 P\Big( {{\sum\limits_{x \in {\Phi _a}}N_{{\rm{tx}}} {{R}\left( {\left. {{\Omega _{{r_r}}},j} \right|\left\| {\bf{x}} \right\|} \right)} } } \Big),
\label{gen_po}
\end{align}
where 
\begin{align}
&R\left( {\left. {{\Omega _{{r_r}}},j} \right|\left\| {\bf{x}} \right\|} \right) \nonumber  \\ &\hspace{0.5cm}= \sum\limits_{i = 1}^j {{b_i}{F^{{\rm{FA}}}}\left( {\left. {{\Omega _{{r_r}}},\left( {j - i} \right){T_b},(j - i + 1){T_b}} \right|\left\| {\bf{x}} \right\|} \right)}, 
\label{R_FA}
\end{align}
for the absorbing receiver, and
\begin{align}
R\left( {\left. {{\Omega _{{r_r}}},j} \right|\left\| {\bf{x}} \right\|} \right)  = \sum\limits_{i = 1}^j {{b_i}{F^{{\rm{PS}}}}\left( {\left. {{\Omega _{{r_r}}},(j - i + 1){T_b}} \right|\left\| {\bf{x}} \right\|} \right)}, 
\label{R_PS}
\end{align}
for the passive receiver. 

In \eqref{R_FA} and \eqref{R_PS}, ${{F^{\rm{FA}}}\left( {\left. {{\Omega _{{r_r}}},\left( {j - i} \right){T_b},(j-i+1){T_b}} \right|{\left\| \textbf{x} \right\|}} \right)}$  and ${{F^{\rm{PS}}}\left( {\left. {{\Omega _{{r_r}}},\left( {j - i + 1} \right){T_b}} \right|{\left\| \textbf{x} \right\|}} \right)}$  are given in 
\eqref{net_deg} and \eqref{PS_nonuniform_der_r0}, respectively.

Compared with the instantaneous receiver observations of a point-to-point system, the instantaneous receiver observations of a large-scale molecular communication system need to account for the statistics of random molecule arrivals from many randomly-placed transmitters.
Based on \eqref{gen_po}, with the fixed threshold-based demodulation, the bit error probability  of  the $j$th randomly-transmitted bit  is derived in the following theorem.

\begin{theorem}
The bit error probability  of the large-scale  molecular communication system  in the $j$th bit   is derived as 
\begin{align}
{P_e}\left[ j \right] = &{P_1}{P_e}\left[ {{{\hat b}_j} = 0\left| {{b_j} = 1,{b_{1:j - 1}}} \right.} \right]
\nonumber\\&  + {P_0}{P_e}\left[ {{{\hat b}_j} = 1\left| {{b_j} = 0,{b_{1:j - 1}}} \right.} \right], \label{overallerror}
\end{align}
where 
\begin{align} 
	&{P_e}\left[ {{{\hat b}_j} = 0\left| {{b_j} = 1,{b_{1:j - 1}}} \right.} \right] \approx
	\nonumber\\ &  \exp \left\{ { - 4\pi{\lambda _a}\int_{{r_r}}^\infty  {\left( {1 - \exp \left\{ { - {N_{{\rm{tx}}}}R\left( {\left. {{\Omega _{{r_r}}},j} \right|r} \right)} \right\}} \right)}  {r^2}\diff r} \right\}
	\nonumber \\& \times\Bigg[ {1 + \sum\limits_{n = 1}^{{N_{{\rm{th}}}} - 1} {\sum\limits^{ n} {\frac{1}{{\prod\limits_{k = 1}^n {{n_k}!k{!^{{n_k}}}} }}\prod\limits_{k = 1}^n {\Big[ { - 4\pi {\lambda _a}\times } \Big.} } } } \Bigg.
	\nonumber \\&\Bigg. {{{\Big. \int_{{r_r}}^\infty  {{{{\left( {{N_{{\rm{tx}}}}R\left( {\left. {{\Omega _{{r_r}}},j} \right|r} \right)} \right)}^k}} }{ {\exp \left\{ { - {N_{{\rm{tx}}}}R\left( {\left. {{\Omega _{{r_r}}},j} \right|r} \right)} \right\}} {r^2}\diff r} \Big]}^{{n_k}}}} \Bigg]
,\label{BER_final}
\end{align}
and
\begin{align} 
	&{P_e}\left[ {{{\hat b}_j} = 1\left| {{b_j} = 0,{b_{1:j - 1}}} \right.} \right] \approx
	\nonumber\\ &  1-\exp \left\{ { - 4\pi {\lambda _a}\int_{{r_r}}^\infty  {\left( {1 - \exp \left\{ { - {N_{{\rm{tx}}}}R\left( {\left. {{\Omega _{{r_r}}},j} \right|r} \right)} \right\}} \right)} {r^2}\diff r} \right\}
	\nonumber \\& \times\Bigg[ {1 + \sum\limits_{n = 1}^{{N_{{\rm{th}}}} - 1} {\sum\limits^{ n}  {\frac{1}{{\prod\limits_{k = 1}^n {{n_k}!k{!^{{n_k}}}} }}\prod\limits_{k = 1}^n {\Big[ { - 4\pi{\lambda _a\times} } \Big.} } } } \Bigg.
	\nonumber \\&\Bigg. {{{\Big. \int_{{r_r}}^\infty  {{{{\left( {{N_{{\rm{tx}}}}R\left( {\left. {{\Omega _{{r_r}}},j} \right|r} \right)} \right)}^k}} }{ {\exp \left\{ { - {N_{{\rm{tx}}}}R\left( {\left. {{\Omega _{{r_r}}},j} \right|r} \right)} \right\}}  {r^2}\diff r} \Big]}^{{n_k}}}} \Bigg]
,\label{BER_final2}
\end{align}
 the summation $\sum\limits^n$ is over all n-tuples of nonegative integers ($n_1,...,n_n$) satisfying the constraint $1 \cdot {n_1} + 2 \cdot {n_2} + \cdots + k \cdot {n_k} +  \cdots +n \cdot {n_n} = n$, ${b_{1:j - 1}}$ is the bit sequence from the  first bit to the $(j-1)$th bit, ${{\hat b}_j}$ is the detected $j$th bit, and $P_1$ and $P_0$ denote the probability of sending bit-1 and bit-0, respectively. In \eqref{BER_final} and \eqref{BER_final2},  $R\left( {\left. {{\Omega _{{r_r}}},j} \right|r} \right)$ is given in \eqref{R_FA} for the absorbing receiver  and \eqref{R_PS} for the passive receiver, respectively.
\end{theorem}
\begin{proof}
See Appendix D.
\end{proof}

The results in Eq. \eqref{BER_final} and Eq. \eqref{BER_final2} of Theorem 1 have combinatorial complexity with multiple sums and products.  In order to gain insight on the impact of the system parameters (except  $N_{\rm{th}}$) on the derived bit error probability, 
we  present a simple expression in the following lemma for the $j$th bit error probability when the detection threshold $N_{\rm{th}}$ equals 1.
\begin{lemma}
With $N_{\rm{th}}=1$, the $j$th bit error probability  of the large-scale molecular communication system with molecule degradation   is given by \eqref{overallerror} with
%\begin{align}
%{P_e}\left[ j \right] = &{P_1}{P_e}\left[ {{{\hat b}_j} = 0\left| {{b_j} = 1,{b_{1:j - 1}}} \right.} \right]
%\nonumber\\&  + {P_0}{P_e}\left[ {{{\hat b}_j} = 1\left| {{b_j} = 0,{b_{1:j - 1}}} \right.} \right], \label{overallerror}
%\end{align}
%where 
\begin{align} 
	&{P_e}\left[ {{{\hat b}_j} = 0\left| {{b_j} = 1,{b_{1:j - 1}}} \right.} \right] \approx
\nonumber\\ &\hspace{0.2cm}
  \exp \left\{ { - {\lambda _a}\int_{{r_r}}^\infty  {\left( {1 - \exp \left\{ { - {N_{{\rm{tx}}}}R\left( {\left. {{\Omega _{{r_r}}},j} \right|r} \right)} \right\}} \right)} 4\pi {r^2}\diff r} \right\},
 \label{fina}
\end{align}
and
\begin{align} 
	&{P_e}\left[ {{{\hat b}_j} = 1\left| {{b_j} = 0,{b_{1:j - 1}}} \right.} \right] \approx
\nonumber\\ &\hspace{0.2cm}
  1-\exp \left\{ { - {\lambda _a}\int_{{r_r}}^\infty  {\left( {1 - \exp \left\{ { - {N_{{\rm{tx}}}}R\left( {\left. {{\Omega _{{r_r}}},j} \right|r} \right)} \right\}} \right)} 4\pi {r^2}\diff r} \right\}
 \label{fina2}.
\end{align}
 In \eqref{fina} and \eqref{fina2},  $R\left( {\left. {{\Omega _{{r_r}}},j} \right|r} \right)$ is given in \eqref{R_FA} for the absorbing receiver,  and \eqref{R_PS} for the passive receiver, respectively.
\end{lemma}
\begin{proof}
See Appendix E.
\end{proof}

To simplify further, we  present the \emph{single} bit error probability (without ISI) of the large-scale molecular communication system \emph{without} molecule degradation at the absorbing receiver with $N_{\rm{th}}=1$ and $k_d =0$ as 
\begin{align} 
	{P_e} \left[ {{{\hat b}_1} = 0\left| {{b_1} = 1} \right.} \right] &\approx \exp \left\{ { - 4\pi{\lambda _a}\int_{{r_r}}^\infty  { {r^2}} } \right.
\nonumber\\ & \hspace{-1.7cm}
 \left. {\Big(1 - \exp \Big\{  - {N_{{\rm{tx}}}}\frac{{{r_r}}}{r}{\rm{erfc}}\Big\{ \frac{{r - {r_r}}}{{\sqrt {4D{T_b}} }}\Big\} \Big\} \Big)\diff r} \right\}.
 \label{fina_ab}
\end{align}

We see that the single bit error probability of the absorbing receiver improves by increasing the diffusion coefficient, the number of transmit molecules, or  the density of  transmitters.  This is because with a single bit-1 transmitted at all the  transmitters, no ISI needs to be considered and so a higher  peak value of net number of absorbed molecules  results in a   better  bit error probability.
%The bit error probability 

\begin{figure}[!tb]
	\centering
	\includegraphics[width=\linewidth]{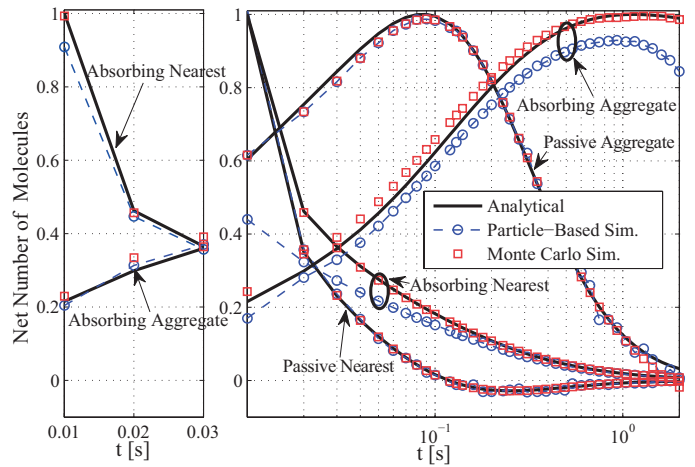}
	\caption{Net number of observed molecules at the receiver as a function of time. All curves are scaled by the maximum value of the analytical curves in the right subplot.}
	\label{fig_accord_normalized}
\end{figure}

\section{Numerical and Simulation Results}

{ Throughout this section, we focus on Monte Carlo approaches to conduct the simulations, which we also compare with particle-based simulations. We consider two types of Monte Carlo simulations. Both types use a HPPP to generate the locations of the transmitters. In the first type, which we use in Figs. 2--5, observations in each realization are ``simulated'' by adding the \emph{expected} observation from every transmitter at the sampling time in \eqref{FA_uo} and \eqref{PS_all}. In the second type, which we use in Figs. 6--9, observations in each realization are ``simulated'' by \emph{drawing from the Poisson distribution} as in \eqref{net_FA_po}  and \eqref{cur_PS_po},  whose mean is the sum of the observations expected from every transmitter at the sampling time. The second type generates distributions of individual observations in order to measure the bit error probability. }

{In  this section, we first validate the Monte Carlo approaches by comparing with particle-based simulations and our analytical results for the net number of molecules at the receiver. Due to the extensive computational demands to simulate large molecular communication environments with a particle-based approach, we then rely on Monte Carlo simulations for further verification of the channel impulse responses and the bit error performance.} In all figures of this section,  we set  $r_r = 5\,\micron$. 
\begin{table}[!tb]
	\centering
	\caption{ The simulation parameters and scaling values applied in Fig.~\ref{fig_accord_normalized}.}
	{\renewcommand{\arraystretch}{0.9}
		\begin{tabular}{|c|c|c|c|c|}
			\hline
			 Transmitter & Receiver & Realizations & Time  & Scaling \\
			&  &  & Step [s] &Value \\ \hline
			Nearest & Passive & $10^4$ & $10^{-2}$ & 149.57 \\ \hline
			Nearest & Absorbing & $10^4$ & $10^{-2}$ & 354.52 \\ \hline
			Aggregate & Passive & $10^4$ & $10^{-2}$ & 9.252 \\ \hline
			Aggregate & Absorbing  & $10^3$ & $10^{-3}$ & 59.42 \\ \hline
		\end{tabular}
	}
	\label{table_accord}
\end{table}
%In Figs. 2-4, we set
%
%In this section, we examine the bit error probability and the expected number of molecules observed at the absorbing receiver and the passive receiver due to joint transmission  at all  point transmitters. 

In Figs.~\ref{fig_accord_normalized}, \ref{Fig2},  and \ref{Fig3}, we set $N_{\rm{tx}}=10^4$, and  $k_d = 0$ to focus on  normal diffusion without molecule degradation in a large-scale  system. The analytical curves of the expected  number of   molecules  absorbed at the absorbing receiver 
due to all  transmitters,  the nearest transmitter, and the other  transmitters are plotted using  \eqref{Total_FA_final}, \eqref{E_u_FA_PS},  and \eqref{E_I_PS}, and are abbreviated as “Absorbing All”, “Absorbing Nearest”, and “Absorbing Aggregate”, respectively. The analytical curves of the expected number of  molecules observed inside   the passive receiver due to all  transmitters, the nearest  transmitter, the other   transmitters  are plotted  using \eqref{Total_PS_final_t}, \eqref{E_u_FA_PS},  and \eqref{E_I_PS}, and are abbreviated as “Passive All”, “Passive Nearest”, and “Passive Aggregate”, respectively. The analytical curves and the simulations are occasionally abbreviated as “Anal.” and “Sim.”, respectively.  %In Fig. \ref{Fig3}, we also plot the analytical curves of the SIRs at the absorbing receiver and the passive receiver  using Eqs. \eqref{SIR_FA} and \eqref{SIR_PS}, which are abbreviated as   “Absorbing SIR” and “Passive SIR”, respectively.

\subsection{Validation of Simulation Approaches}

{ The Monte Carlo approaches assume  that the channel response for a single transmitter is correct. We check this assumption by comparing the first Monte Carlo approach with particle-based simulations generated using the simulation algorithm in \cite{Yansha16} and the AcCoRD simulator (Actor-based Communication via Reaction-Diffusion) \cite{Noel2017}. In the first Monte Carlo approach, every realization is simulated by calculating the net number of molecules due to each transmitter using \eqref{FA_uo} and \eqref{PS_all} for the absorbing and passive receivers, respectively. In the particle-based approach, observations in each realization are ``simulated'' by placing individual molecules at each transmitter, moving each molecule by Brownian motion, and checking whether each molecule diffused into the passive receiver or was absorbed by the absorbing receiver. AcCoRD simulations are defined by configuration files; here, each configuration file listed the transmitter locations as specified by the current permutation of the HPPP, and each transmitter permutation was simulated at least 10 times.}

The simulation approaches are compared in Fig.  \ref{fig_accord_normalized}, where we set  $D = 80\times10^{-12} \frac{\meter^2}{\second}$ and assume that the transmitters are placed up to $R_a =50\mu$m from the center of the receiver at a density of $\lambda_a =10^{-4}$ transmitters per $\mu$m$^3$ (i.e., 52 transmitters on average, including the exclusion of the receiver volume). The receiver takes samples every $T_{ss} = 0.01$ s and calculates the \emph{net change} in the number of observed molecules between samples. The default simulation time step is also $0.01$ s. Unless otherwise noted, all simulation results were averaged over $10^4$ transmitter location permutations, as shown in Table II.

%In Fig. \ref{fig_accord_normalized}, we set   $D = 80\times10^{-12} \frac{\meter^2}{\second}$, and assume that the transmitters are placed up to $R_a = 50\,\micron$ from the center of the receiver at a density of $\lambda_a = 10^{-4}$ transmitters per $\micron^3$ (i.e.,  52 transmitters on average, including the subtraction of the receiver volume).
% The receiver takes samples every $T_{ss} = 0.01\; \second$ and calculates the net change in the number of observed molecules between samples. The default simulation time step  is also $0.01$ $\rm{\second}$. Unless otherwise noted, all simulation results were averaged over $10^4$ transmitter location permutations, with each permutation simulated at least 10 times.

In Fig. \ref{fig_accord_normalized}, we verify the analytical expressions for the expected net number of molecules observed during $[t,t+T_{ss}]$ at both receivers in  \eqref{E_u_FA_PS},  and  \eqref{E_I_PS}  by comparing with the particle-based simulations and the Monte Carlo simulations. 
%The  particle-based  simulations were performed by tracking the progress
%of individual particles to obtain the net number of  observed molecules     during $[t,t+T_{ss}]$ using the Algorithm 1 in \cite{Yansha16} and the AcCoRD simulator (Actor-based Communication via Reaction-Diffusion) \cite{Noel2017}.  The  Monte Carlo simulations  were performed   by  averaging the  expected net number of observed molecules  due to  all  transmitters with  randomly-generated locations, as calculated based on  \eqref{FA_uo} and \eqref{PS_all}, over $10^4$ realizations.
In the right subplot of Fig.~\ref{fig_accord_normalized}, we compare passive and absorbing receivers and observe the expected  net number of observed molecules during $[t,t+T_{ss}]$ due to the nearest transmitter and due to the other   transmitters. In the left subplot of Fig.~\ref{fig_accord_normalized}, we lower the simulation time step  to $10^{-4}\,\second$ for the first few samples of the two absorbing receiver cases, in order to demonstrate the corresponding improvement in accuracy\footnote{ Only the data points at intervals of 10$^{-2}$s are presented   in the left subplot of Fig. \ref{fig_accord_normalized} to avoid crowded  markers.}. All curves in both subplots are scaled by the maximum value of the corresponding analytical curve in the right subplot; the scaling values and other simulation parameters are summarized in Table~\ref{table_accord}.

\subsubsection{Particle-Based  Simulation Validation}
Overall, there is good agreement between the analytical curves and the particle-based simulations in the right subplot of Fig.~\ref{fig_accord_normalized}. The analytical results for the net number of molecules observed  inside the passive receiver  during $[t,t+T_{ss}]$ due to  the nearest transmitter are highly accurate, and even captures the net loss of molecules observed after $t = 0.1\,\second$. { The particle-based simulation of the ``Passive Aggregate'' case also becomes noisier with increasing $t$ as the normalized net number of molecules goes below 0.3, which is due to the very low number of molecules observed (the scaling factor in this case is only 9.525; see Table~\ref{table_accord}) and can be improved by averaging over more realizations.}
{ Both simulation approaches slightly underestimate the analytical curve in the ``Passive Aggregate'' case for $t < 0.1$ s, due to the constraint on the placement of transmitters to within a radius of $R_a$ = 50 $\mu$m (which we relax in later figures once we do not include particle-based simulations).}

\begin{figure}[!tb]
	\centering
	\includegraphics[width=\linewidth]{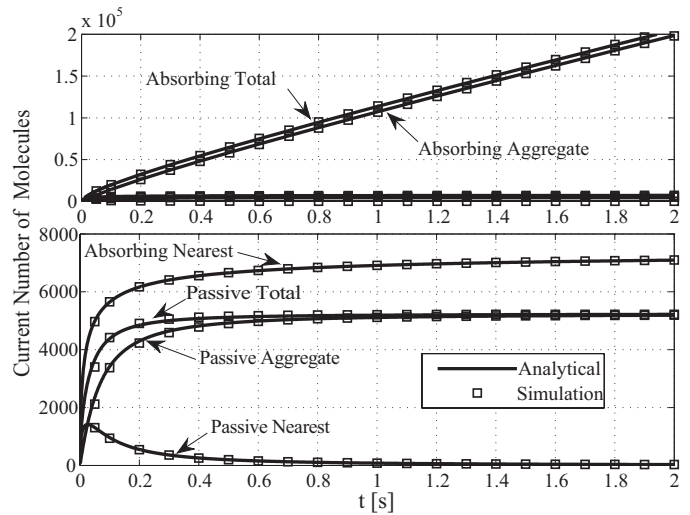}
	\caption{Expected number of  molecules observed at the receiver  as a function of time. }
	\label{Fig2}
\end{figure} 

 There is less agreement between the particle-based simulations and the analytical expressions for the absorbing receiver, and this is primarily due to the large simulation time step (even though we used a smaller time step for the aggregate transmitter case in the right subplot; see Table II). To demonstrate the impact of the time step, the left subplot shows much better agreement for the absorbing receiver model by lowering the time step to $10^{-4}\,\second$. This improvement is especially true in the case of the nearest transmitter, as there is significant deviation between the particle-based simulation and the analytical expression for very early times in the right subplot.

%even though the simulations restricted the placement of transmitters to the maximum distance $R = 50\,\micron$

\subsubsection{Monte Carlo Simulation Validation}
There is a good match between the analytical curves and the Monte Carlo simulations for the  net number of  molecules observed at both types of receivers during $[t,t+T_{ss}]$ due to the nearest transmitter, which can be attributed to the large number of  molecules  (as shown in Table~\ref{table_accord}) and  the small value of the shortest distance between the transmitter and the receiver  compared with $R_a = 50\,\micron$. There is slight deviation in the Monte Carlo simulations for the expected number of  molecules observed at both types of receivers due to the other transmitters, and this is primarily due to
  the restricted placement of transmitters to the maximum distance $R_a = 50\,\micron$.  In Figs. \ref{Fig2} and \ref{Fig3},  better agreement between the analytical curves and Monte Carlo simulation is achieved by increasing  the maximum placement distance  $R_a$.

\begin{figure}[!tb]
	\centering
	\includegraphics[width=\linewidth]{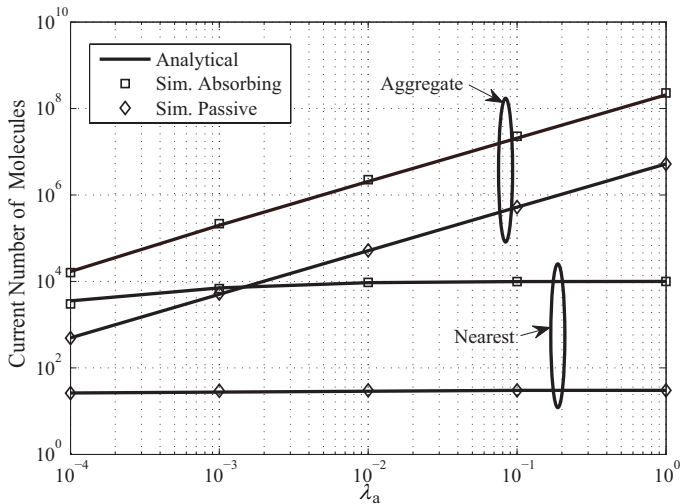}
	\caption{Expected number of  molecules observed at the receiver at time $t=2\;\;\second$ as a function of the density of transmitters. }
	\label{Fig3}
\end{figure}

Due to the extensive computational demands to simulate such large molecular communication environments, we assume that the particle-based simulations have sufficiently verified the analytical models. The remaining simulation results in the rest of the figures  are only  generated via Monte Carlo simulation.

\subsection{Channel Impulse Response Evaluation}
From Fig.~\ref{fig_accord_normalized} and  the scaling values in Table II,  we see that the expected net number of  molecules observed at the absorbing receiver  is much larger than that inside the passive receiver, since every molecule arriving at the absorbing receiver is permanently absorbed.  We also notice that  the expected  net number of observed molecules due to the nearest transmitter  is much larger than that due to the  other transmitters, which may be due to a relatively low transmitter density.

%Interestingly, the concurrent single pulse transmission by the transmitters at time $t=0$ results in a longer and stronger \emph{net} number of observed molecules at the absorbing receiver than that at the passive receiver. If the demodulation is based on the number of observed molecules during each bit interval, the longer channel response at the absorbing receiver may contribute to higher ISI  than at a passive receiver for the same bit interval,  whereas its stronger channel response may benefit  signal detection in the intended interval.

Figs.~\ref{Fig2} and \ref{Fig3} plot the expected number of  molecules currently observed at the absorbing receiver and the passive receiver at time $t$ rather than their net change during each sampling interval.
In Figs.~\ref{Fig2} and \ref{Fig3}, we set the parameters: $D = 120\times10^{-12} \frac{\meter^2}{\second}$, $R_a = 100\,\micron$, and $T_{ss} = 0.1\; \second$.
 We set the density of  transmitters as $\lambda_a = 10^{-3} /\micron^{3}$.  
As shown in the lower subplot of  Fig.~\ref{Fig2},  even though the point transmitters have random locations,
the   channel responses of the receivers due to the nearest transmitter in this large-scale molecular communication system  are consistent with those observed at the absorbing receiver in   \cite[Fig. 4]{Yansha16} and the passive receiver in  \cite[Fig. 2]{Llatser13} and \cite[Fig. 1]{noel2014improving}   for a point-to-point molecular communication system.

\begin{figure}[!tb]
	\centering
	\includegraphics[width=\linewidth]{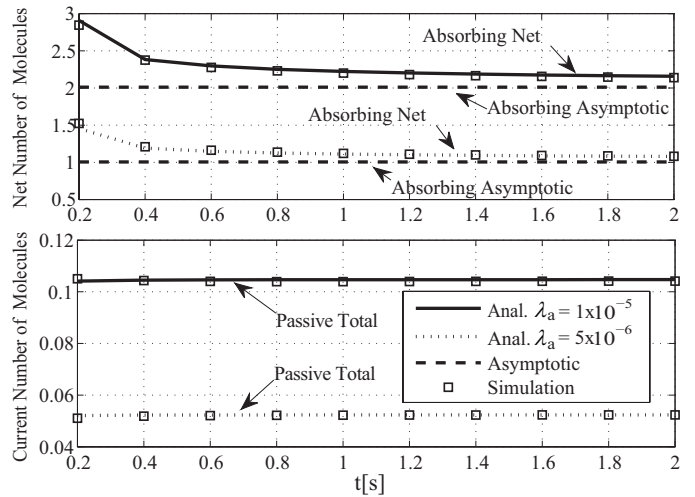}
	\caption{Expected net number of  molecules observed at the receiver  as a function of time. }
	\label{Fig4}
\end{figure} 
 
In Fig.~\ref{Fig2}, we notice that the  expected   number of  molecules currently observed at time $t$ due to all  transmitters is dominated by the other  transmitters, rather than the nearest transmitter, which is due to the increased number of molecules received from the other  transmitters with the higher density of transmitters compared to that in  Fig.~\ref{fig_accord_normalized}.
Furthermore, as we might expect, the  expected  number of  molecules currently observed inside the passive receiver at time $t$ stabilizes after $t=0.8\,\second$, whereas that at the absorbing receiver eventually increases linearly with increasing time. This reveals the potential differences in appropriate demodulation  design for these two types of receiver. More specifically, unlike the demodulation for passive receiver, demodulation using the  number of  molecules currently absorbed by the absorbing receiver  is not a suitable design, since it cannot have a single optimal threshold.

Fig.~\ref{Fig3}  plots the expected number of  molecules observed at the absorbing receiver and the passive receiver at $t=2\,$s versus the density of  transmitters $\lambda_a$.   With the increase of $\lambda_a$, 
 the  number of observed molecules due to the other  transmitters increases,  whereas the  number of observed molecules due to the nearest transmitter remains almost unchanged. More importantly, the dominant effect of   the other transmitters on the number of observed molecules  becomes more obvious as 
  $\lambda_a$ increases. 

\subsection{Demodulation Criteria and Single Bit Error Performance}

From Figs.~\ref{Fig2} and~\ref{Fig3},  the \emph{current} number of absorbed molecules increases with increasing time and transmitter density, thus demodulation based on the \emph{current} number of  molecules absorbed by the absorbing receiver will require an increasing demodulation threshold for larger $t$ and $\lambda_a$.  Hence, in our model, the demodulation of the absorbing receiver is based on the \emph{net} number of absorbed molecules, whereas the demodulation of the passive receiver is based on the \emph{current} number of  molecules  observed at the receiver.  In Figs. \ref{Fig4} and \ref{Fig5}, we set $N_{\rm{tx}} = 20$, $k_d = 0$, $T_b =0.2$ s,  $R_a = 100\,\micron$, $D = 80\times10^{-11} \frac{\meter^2}{\second}$, and with only a single bit-1 transmitted at $t=0$, i.e.,  the transmit bit sequence is [1 0 0 0 $ \ldots$ ].

\begin{figure}[!tb]
	\centering
	\includegraphics[width=\linewidth]{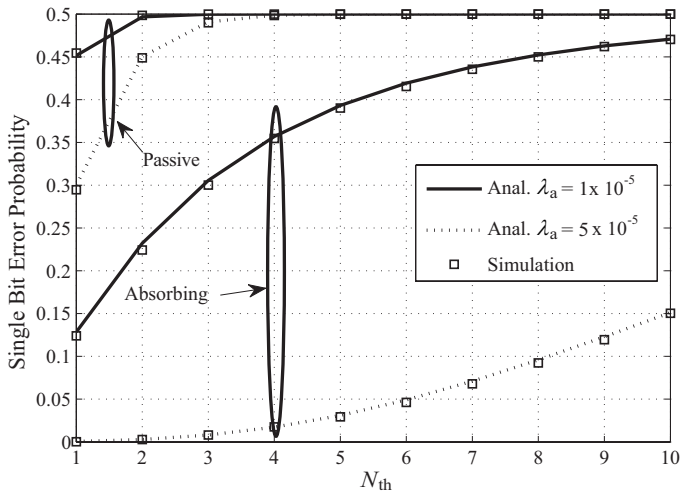}
	\caption{Single bit error probability as a function of threshold. }
	\label{Fig5}
\end{figure}

 Fig. \ref{Fig4} plots the net number of   molecules absorbed by the absorbing receiver during one bit interval $T_b$ in the upper subfigure, and the number of observed molecules at the passive receiver at the end of each bit interval $T_b$ in the lower subfigure, each with different transmitter densities. We  also plot the asymptotic net number of absorbed  molecules using \eqref{Total_FA_final_asy} with a dashed line. 
We see that the  net number of molecules absorbed by the absorbing receiver during each bit interval decreases as  time increases, and converges to the asymptotic value.
The number of observed molecules inside the passive receiver at the end of  every bit interval remains comparable as  time increases, which suggests that  taking multiple samples of the number of observed molecules at different times in one bit interval may not greatly improve the detection reliability.
  For both receivers, the ISI is not small compared with the observation in the first bit interval,  which demonstrates the  high ISI in the large-scale  molecular communication system.

In Fig. \ref{Fig5},  we start using the second Monte Carlo approach for simulations in order to generate distributions of observations, and we plot the single bit error probability of both receivers using \eqref{overallerror}, in order to focus on the impact of multiple transmitters with no ISI impairment.  We notice that the single bit error probability at both receivers improves with increasing $\lambda_a$, which is due to the increased number of molecules absorbed by the absorbing receiver  during $t \in [0, T_b]$, and the increased number of observed molecules inside the passive receiver at $t=T_b$ as seen in Fig. \ref{Fig4}. Another interesting observation is that the single bit error probability of the passive receiver is much worse than that of the absorbing receiver, which is due to the lower number of observed molecules at the passive receiver than that at the absorbing receiver. 
Clearly, the two receivers need different demodulation thresholds.
\begin{figure}[!tb]
	\centering
	\vspace{-0.12cm}
	\includegraphics[width=\linewidth]{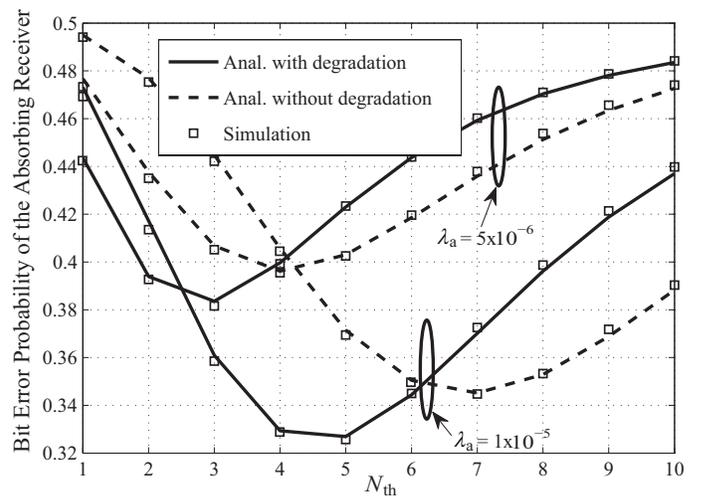}
	\caption{Bit error probability of the absorbing  receiver as a function of threshold. }
	\label{Fig6}
\end{figure}

\subsection{Multiple Bits Error Performance }

 Figs. \ref{Fig6} and  \ref{Fig7}  plot the bit error probabilities of the absorbing receiver and that of the passive receiver in the proposed large-scale molecular communication system, respectively, both 
 with $(k_d=0.8$ ${\second}^{-1})$ or without $(k_d=0$ ${\second}^{-1})$ molecule degradation.
Fig. \ref{Fig8}  compares the bit error probabilities of the absorbing receiver and the passive receiver in the proposed large-scale molecular communication system under molecule degradation $(k_d=0.8$ ${\second}^{-1})$  using DFD, with that  using the simple detector. 
  In Figs. \ref{Fig6},  \ref{Fig7}, and \ref{Fig8}, we set the parameters: $T_b =0.2$ s,  $R_a = 100\,\micron$, and $D = 80\times10^{-11} \frac{\meter^2}{\second}$ with a 5 bit sequence  transmitted by all transmitters, where the first four bits are set as [1 0 1 0] .
 We set $N_{\rm{tx}} =20$ in Fig. \ref{Fig6}, $N_{\rm{tx}} =300$ in Fig. \ref{Fig7}, and $N_{\rm{tx}} =10^4$ in Fig. \ref{Fig8}.

  In Figs. \ref{Fig6} and \ref{Fig7},  we see a good match between the analytical results in  \eqref{overallerror}  and the simulations, which demonstrates the correctness of our derivations. We observe that the minimum bit error probability  improves with increasing the density of the  transmitters.
  We also see that the minimum bit error probability  can be improved by introducing   molecule degradation. This can be explained by the fact that many molecules, especially those released far from the receiver, degrade before they reach the receiver, and this reduces the ISI effect.  However, the bit error probability with molecule degradation is not always better than without degradation for a given decision threshold, which can be attributed to the fact that the degradation not only reduces the ISI, but also lowers the strength of the intended signal.
  
    In both figures, we notice that the minimum bit error probability is still not low enough for reliable transmission, even though it can be potentially improved by increasing $N_{\rm{tx}}$. This is because with multiple transmitted bits, 
the ISI will accumulate and  keep growing with every transmit bit-1.  These observations reveal   that the demodulation threshold at each bit should increase with the number of transmit bits, instead of being fixed.
\begin{figure}[!tb]
	\centering
	\includegraphics[width=\linewidth]{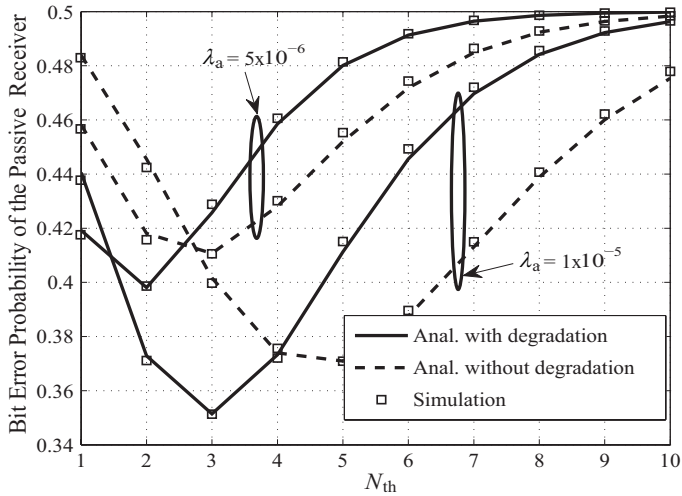}
	\caption{Bit error probability of the passive  receiver as a function of threshold. }
	\label{Fig7}
\end{figure} 
\begin{figure}[!tb]
	\centering
	\includegraphics[width=\linewidth]{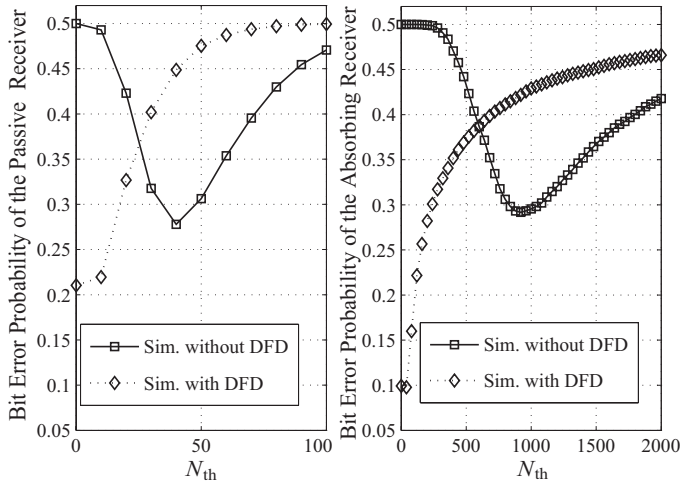}
	\caption{Bit error probability of  receivers as a function of threshold. }
	\label{Fig8}
\end{figure} 

We now consider the DFD  at both receivers to show its potential benefits in improving the bit error probability.
Fig. \ref{Fig8}  compares the bit error probability of both receivers having molecule degradation during diffusion and DFD during detection with that without DFD during detection using Monte Carlo simulation, where the passive receiver is capable of  subtracting the  \emph{current} observation in one previous bit interval $ N\left[ j-1 \right]$ from that in the current bit interval $ N\left[ j \right]$,  and the absorbing receiver is  capable of  subtracting the  \emph{net} observation in one previous bit interval $ N\left[ j-1 \right]$ from that in the current bit interval $ N\left[ j \right]$ for the demodulation of the $j$th bit.  With DFD, the $j$th bit is decoded based on if $ N\left[ j \right]- N\left[ j-1 \right]>N_{\rm{th}}$ or not.
By doing so, the accumulated ISI due to the previous bits is mitigated artificially during the demodulation process.  We set $\lambda_a = 5\times10^{-6}  /\micron^{3}$.
With the help of DFD, we see that the minimum  bit error probability of both receivers can be  improved for the proposed system.

\section{Conclusions and Future Work}
In this paper, we provided a general model for the collective signal  modelling in a large-scale molecular communication system with or without degradation using stochastic geometry.
The  collective signal strength at a fully absorbing receiver and a passive receiver is modelled and explicitly characterized. We derived  tractable  expressions for the expected  number of observed molecules at the fully absorbing receiver and the passive receiver, which were shown to increase with transmitter density.
We also  derived  analytical expressions for the  bit error probabilities at both receivers with a simple detector taking one sample per bit, and the minimum bit error probabilities  were shown to improve with the help of degradation. 
 The analytical model presented in this paper can also be applied for the performance evaluation of other types of receiver (e.g., partially absorbing, reversible adsorption receiver, ligand-binding receiver)   in a large-scale  molecular communication system by substituting its corresponding channel response.

\appendices
\numberwithin{equation}{section}

\section{ Proof of Proposition 1}\label{proposition1}
According to \cite{Haenggi05}, the probability of finding $k$ nodes in a bounded Borel $A \subset {{\mathbb{R}}^m}$ in a homogeneous $m$-dimensional Poisson point process of intensity $\lambda$ is given by
\begin{align}
\Pr \left( {M = k} \right) = {e^{ - \lambda_a \mu \left( A \right)}}\frac{{{{\left( {\lambda_a \mu \left( A \right)} \right)}^k}}}{{k!}},
\end{align}
where $M$ is the Poisson random variable,  and $\mu \left( A \right)$ is the standard Lebesgue measure of $A$.

Thus, the probability of finding zero nodes in a bounded Borel $A \subset {{\mathbb{R}}^3}$ in a homogeneous $3$D Poisson point process of intensity $\lambda_a$ is obtained as
\begin{align}
\Pr \left( {M = 0} \right) = {e^{ - \lambda_a \mu \left( A \right)}},
\end{align}
where  $\mu \left( A \right)= {\frac{4}{3}\pi {x^3} - \frac{4}{3}\pi {r_r}^3} $, and $x$ is the radius of the bounded ball.

Using ${f_{\left\| {{x^*}} \right\|}}(x) =  - \frac{{d\Pr \left( {N = 0} \right)}}{{dx}}$, we prove \eqref{shortest_dist}.

\section{  Proof of Theorem 1}\label{theorem1}
Based on \eqref{total_FA} and \eqref{PS_all}, we can write the expected net number of  molecules observed at the  receiver as
\begin{align}
{\mathbb{E}}&\left\{ N_{{\rm{all}}}\left( { {{\Omega _{{r_r}}},t,t + {T_{ss}}} } \right) \right\}  =  {\mathbb{E}}\Big\{ \sum\limits_{x \in\Phi_a } {{N_{\rm{tx}}}\Phi \left( r \right)}  \Big\}, \label{proof1_1_a}
\end{align}
where
\begin{align}
\Phi \left( r \right) = {F^{{\rm{FA}}}}\left( {\left. {{\Omega _{{r_r}}},t,t + {T_{ss}}} \right|r} \right), \label{phi_FA_a}
\end{align}
 for the absorbing receiver, and
\begin{align}
\Phi \left( r \right) = {F^{{\rm{PS}}}}\left( {\left. {{\Omega _{{r_r}}},t + {T_{ss}}} \right|r} \right) - {F^{{\rm{PS}}}}\left( {\left. {{\Omega _{{r_r}}},t} \right|r} \right), \label{phi_PS_a}
\end{align}
for the passive receiver.

According to the  Campbell’s theorem in 3D space, the mean of the random sum of a  point  process $\Phi_a$ on ${\mathbb{R}}^3$ and ${{N_{\rm{tx}}}\Phi \left( r \right)}$ is given as \cite[Eq. (1.18)]{baccelli2009stochastic}
\begin{align}
{\mathbb{E}}\left\{ N_{{\rm{all}}}\left( { {{\Omega _{{r_r}}},t,t + {T_{b}}} } \right) \right\}  &=  \int_{{{\mathbb{R}}^3}} {\left[ {N_{{\rm{tx}}}\Phi \left( r \right)} \right]{\lambda _a}\diff x} 
\nonumber\\& \hspace{-0.5cm}= {\lambda _a}\int_{{r_r}}^\infty  {\left[ {N_{{\rm{tx}}}\Phi \left( r \right)} \right]3\frac{{4\pi }}{3}{r^2}\diff  r}  . \label{proof1_2}
\end{align}
 
Thus, we derive 
\begin{align}
{\mathbb{E}}&\left\{ N_{{\rm{all}}}\left( { {{\Omega _{{r_r}}},t,t + {T_{ss}}} } \right) \right\}  =  4\pi {\lambda _a} {N_{\rm{tx}}^{\rm{FA}}}
\int_{{r_r}}^\infty  {\Phi  \left( r \right)} {r^2}\diff r.
 \label{proof1}
\end{align}
% where 
%\begin{align}
%&  {\Pi_t^{t + {T_{b}}}}\left( r \right)  =  \int_t^{t + {T_{b}}} {\frac{{{r_r}}}{r}\frac{{r - {r_r}}}{{\sqrt {4\pi D{x^3}} }}\exp \Big( { - \frac{{{{\left( {r - {r_r}} \right)}^2}}}{{4Dx}}-k_dx} \Big)dx} , \label{phi}
%\end{align}
%which is the fraction of molecules absorbed during any time interval $[t,t + {T_{b}}]$.

\section{  Proof of Lemma 1}\label{lemma1}
With $k_d=0$, we  rewrite  \eqref{proof1} using $z=r-r_r$ as
\begin{align}
&  {\mathbb{E}}\left\{ N_{{\rm{all}}}^{{\rm{FA}}}\left( { {{\Omega _{{r_r}}},t,t + {T_{ss}}} } \right) \right\}  =
\nonumber\\  &  \hspace{0.15cm} \frac{{\sqrt {4\pi } {\lambda _a}{{N_{\rm{tx}}}}{r_r}}}{{\sqrt D }}\int_t^{t + {T_{ss}}} {\int_0^\infty  {z\left( {z + {r_r}} \right)\exp \big( { - \frac{{{z^2}}}{{4Dx}}} \big)\diff z} } \frac{1}{{\sqrt {{x^3}} }}\diff x
\nonumber\\  &  \hspace{0.15cm} = \frac{{\sqrt {4\pi } {\lambda _a}{N_{\rm{tx}}}{r_r}}}{{\sqrt D }}\Bigg[ {\int_t^{t + {T_{ss}}} {\int_0^\infty  {{z^2}\exp \big( { - \frac{{{z^2}}}{{4Dx}}} \big)\diff z} } \frac{1}{{\sqrt {{x^3}} }}\diff x} \Bigg.
\nonumber\\  &  \hspace{0.2cm} \Bigg. { + {r_r}\int_t^{t + {T_{ss}}} {\int_0^\infty  {z\exp \big( { - \frac{{{z^2}}}{{4Dx}}} \big)\diff z} } \frac{1}{{\sqrt {{x^3}} }}\diff x} \Bigg],
 \label{proof21}
\end{align}

With  mathematical manipulations, we simplify \eqref{proof21}  as
\begin{align}
&  {\mathbb{E}}\left\{ N_{{\rm{all}}}^{{\rm{FA}}}\left( { {{\Omega _{{r_r}}},t,t + {T_{ss}}} } \right) \right\}  
%\frac{{\sqrt {4\pi } {\lambda _a} {N_{\rm{tx}}^{\rm{FA}}}{r_r}}}{{\sqrt D }} \nonumber\\  &  \hspace{0.65cm} \int_t^{t + {T_{b}}} {\int_{{r_r}}^\infty  {\left( {r - {r_r}} \right)r\exp \Big( { - \frac{{{{\left( {r - {r_r}} \right)}^2}}}{{4Dx}}} \Big)dr} } \frac{1}{{\sqrt {{t^3}} }}dt
%\nonumber\\ & \hspace{0.15cm}= \frac{{\sqrt {4\pi } {\lambda _a} {N_{\rm{tx}}^{\rm{FA}}}{r_r}}}{{\sqrt D }}\Big[ {\int_T^{T + {T_{ss}}} {\int_0^\infty  {{z^2}\exp \Big( { - \frac{{{z^2}}}{{4Dt}}} \Big)dz} } \frac{1}{{\sqrt {{t^3}} }}dt} \Big.
%\nonumber\\ & \hspace{0.65cm} \Big. { + {r_r}\int_T^{T + {T_{ss}}} {\int_0^\infty  {z\exp \Big( { - \frac{{{z^2}}}{{4Dt}}} \Big)dz} } \frac{1}{{\sqrt {{t^3}} }}dt} \Big]
\nonumber\\& \hspace{0.15cm} = 4\sqrt \pi  {\lambda _a} {N_{\rm{tx}}}{r_r}\Bigg[ {D\sqrt \pi  \int_t^{t + {T_{ss}}} {\diff x}  + \sqrt D {r_r}\int_t^{t + {T_{ss}}} {\frac{1}{{\sqrt x }}\diff x} } \Bigg]. \label{proof2}
\end{align}

Solving \eqref{proof2}, we prove Lemma 1.

\section{  Proof of Theorem 2}\label{theorem2}	
	
	Based on the fact that 
\begin{align} 
	{\left. {\frac{{{\partial ^n}\left( {\exp \left\{ { - N_{{\rm{tx}}}\phi x\tau } \right\}} \right)}}{{\partial {x^n}}}} \right|_{x = {\phi ^{ - 1}}}} = \exp \left\{ { - N_{{\rm{tx}}}\tau } \right\}{\left( { - N_{{\rm{tx}}}\phi \tau } \right)^n},
\end{align}
we rewrite  the error probability for the  transmit bit-1 signal in the $j$th bit as
\begin{align} 
	&{P_e}\left[ {\hat b_{j} = 0\left| {b_j = 1} \right.} \right] = \int_0^\infty  {\exp \left\{ { - {N_{{\rm{tx}}}}\tau } \right\}} {f_{R_{\rm{tot}}^j}}\left( \tau  \right)\diff \tau   +
	\nonumber\\ &  \sum\limits_{n = 1}^{{N_{{\rm{th}}}} - 1} {\frac{1}{{{{\left( { - \phi } \right)}^n}n!}}\int_0^\infty  {{{\left. {\frac{{{\partial ^n}\left( {\exp \left\{ { - {N_{{\rm{tx}}}}\phi x\tau } \right\}} \right)}}{{\partial {x^n}}}} \right|}_{x = {\phi ^{ - 1}}}}} {{f_{R_{\rm{tot}}^j}}\left( \tau  \right)\diff \tau }} 
	\nonumber \\&  = {\mathscr{L}_{R_{\rm{tot}}^j}}\left( {{N_{{\rm{tx}}}}} \right) + {\sum\limits_{n = 1}^{{N_{{\rm{th}}}} - 1} {\frac{1}{{{{\left( { - \phi } \right)}^n}n!}}\left. {\frac{{{\partial ^n}\left[ {{\mathscr{L}_{R_{\rm{tot}}^j}}\left( {{N_{{\rm{tx}}}}\phi x} \right)} \right]}}{{\partial {x^n}}}} \right|} _{_{x = {\phi ^{ - 1}}}}}
, \label{BER2}
\end{align}
where ${{f_{R_{\rm{tot}}^j}}\left( \tau  \right) }$ is the PDF of ${R_{\rm{tot}}^j}$, and ${\mathscr{L}_{R_{\rm{tot}}^j}} \left(  \cdot  \right)$ is the Laplace transform of  ${R_{\rm{tot}}^j}$.

According to \eqref{R_tot_j}, the Laplace transform of ${R_{\rm{tot}}^j}$  can be represented as 
\begin{align} 
	{\mathscr{L}_{R_{\rm{tot}}^j}} \left(  s  \right)= &{\mathbb{E}}\Big[ {\exp \Big\{ { - s{\sum _{{\Phi _a}}}R\left( {\left. {{\Omega _{{r_r}}},j} \right|\left\| {\bf{x}} \right\|} \right)} \Big\}} \Big]
	\nonumber\\ =& {\mathbb{E}}\Big[ {\prod\limits_{{\Phi _a}} {\exp \left\{ { - sR\left( {\left. {{\Omega _{{r_r}}},j} \right|\left\| {\bf{x}} \right\|} \right)} \right\}} } \Big]
\nonumber\\ &\hspace{-1.2cm} =\exp \left\{ { - {\lambda _a}\int_{{{\mathbb{R}}^3}} {\left( {1 - \exp \left\{ { - sR\left( {\left. {{\Omega _{{r_r}}},j} \right|\left\| {\bf{x}} \right\|} \right)} \right\}} \right)\diff\left\| {\bf{x}} \right\|} } \right\}
\nonumber\\ &\hspace{-1.2cm}
 = \exp \left\{ { - {\lambda _a}\int_{{r_r}}^\infty  {\left( {1 - \exp \left\{ { - sR\left( {\left. {{\Omega _{{r_r}}},j} \right|r} \right)} \right\}} \right)} 4\pi {r^2}\diff r} \right\}.
 \label{lap1}
\end{align}

Based on \eqref{lap1} and the Fa$\grave{\text{a}}$
 di Bruno's formula \cite{roman1980formula}, we derive
\begin{align} 
	&{\left. {\frac{{{\partial ^n}\left[ {{\mathscr{L}_{R_{\rm{tot}}^j}}\left( {{N_{{\rm{tx}}}}\phi x} \right)} \right]}}{{\partial {x^n}}}} \right|} _{_{x = {\phi ^{ - 1}}}} =
	\nonumber \\  & \hspace{0.2cm}\exp \left\{ { - {\lambda _a}\int_{{r_r}}^\infty  {\left( {1 - \exp \left\{ { - {N_{{\rm{tx}}}}R\left( {\left. {{\Omega _{{r_r}}},j} \right|r} \right)} \right\}} \right)} 4\pi {r^2}\diff r} \right\}
	\nonumber \\& \hspace{0.2cm}
	\times{\sum^n} {\frac{{n!}}{{\prod\limits_{k = 1}^n {{n_k}!k{!^{{n_k}}}} }}} \prod\limits_{k = 1}^n {\Big[ { - {\lambda _a}\int_{{r_r}}^\infty  {\big[ { - {{\left( { - {N_{{\rm{tx}}}}\phi R\left( {\left. {{\Omega _{{r_r}}},j} \right|r} \right)} \right)}^k}} \big.} } \Big.} 
	\nonumber \\&  \hspace{0.2cm} \times{\Big. {\big. {\exp \left\{ { - {N_{{\rm{tx}}}}R\left( {\left. {{\Omega _{{r_r}}},j} \right|r} \right)} \right\}} \Big]4\pi {r^2}\diff r} \Big]^{{n_k}}}
, \label{BER3}
\end{align}
where the summation $\sum\limits^n $ is over all n-tuples of nonegative integers ($n_1,...,n_n$) satisfying the constraint $1 \cdot {n_1} + 2 \cdot {n_2} + \cdots + k \cdot {n_k} +  \cdots +n \cdot {n_n} = n$.
Noting that $\prod\limits_{k = 1}^n {{{\left( { - \phi } \right)}^{k{n_k}}} = } {\left( { - \phi } \right)^n}$, and substituting   \eqref{lap1} and \eqref{BER3} into \eqref{BER2}, we finally derive \eqref{BER_final}. We can follow a similar method to  derive \eqref{BER_final2}.

\section{  Proof of Lemma 3}\label{lemma3}
With fixed threshold-based demodulation, the error probability  with the  transmit bit-1 signal in the $j$th bit is represented as 
\begin{align}
&{P_e}\left[ {{{\hat b}_j} = 0\left| {{b_j} = 1,{b_{1:j - 1}}} \right.} \right] = \Pr \left[ {N\left[ j \right] < 1} \right]
\nonumber\\  &\approx     {{\mathbb{E}}_{R_{\rm{tot}}^j}}\left\{ {\Pr \left[ {\left. {{P}\left( {{N_{{\rm{tx}}}}R_{\rm{tot}}^j} \right) < 1} \right|R_{\rm{tot}}^j} \right]} \right\}
\nonumber\\  &
 = {{\mathbb{E}}_{R_{\rm{tot}}^j}}\left\{ {\sum\limits_{n = 0}^{0} {\frac{1}{{n!}}\exp \left\{ { - {N_{{\rm{tx}}}}R_{\rm{tot}}^j} \right\}{{\left( {{N_{{\rm{tx}}}}R_{\rm{tot}}^j} \right)}^n}} } \right\}
\nonumber\\  & =
{\mathscr{L}_{R_{\rm{tot}}^j}} \left( {N_{{\rm{tx}}}}  \right),   \label{der55}
\end{align} 
where
\begin{align}
R_{\rm{tot}}^j = \sum\limits_{x \in {\Phi _a}} {R\left( {\left. {{\Omega _{{r_r}}},j} \right|\left\| {\bf{x}} \right\|} \right)} ,  \label{R_tot_j}
\end{align}
	with ${R\left( {\left. {{\Omega _{{r_r}}},j} \right|\left\| {\bf{x}} \right\|} \right)}$ given in \eqref{R_FA} for the absorbing receiver and in \eqref{R_PS} for the passive receiver.
	Substituting \eqref{lap1} into \eqref{der55}, we derive 	  \eqref{fina}.
We can follow a similar method to  derive \eqref{fina2}.

With the threshold-based demodulation, the error probability  for a  transmit bit-1 signal in the $j$th bit is represented as 
\begin{align}
&{P_e}\left[ {\hat b_{j} = 0\left| {b_j = 1} \right.} \right] = \Pr \left[ {N\left[ j \right] < {N_{{\rm{th}}}}} \right]
\nonumber\\  &\approx     {{\mathbb{E}}_{R_{{\rm{tot}}}^j}}\left\{ {\Pr \left[ {\left. {{P}\left( {{N_{{\rm{tx}}}}R_{{\rm{tot}}}^j} \right) < {N_{{\rm{th}}}}} \right|R_{\rm{tot}}^j} \right]} \right\}
\nonumber\\  &
 = {{\mathbb{E}}_{R_{\rm{tot}}^j}}\left\{ {\sum\limits_{n = 0}^{{N_{{\rm{th}}}} - 1} {\frac{1}{{n!}}\exp \left\{ { - {N_{{\rm{tx}}}}R_{\rm{tot}}^j} \right\}{{\left( {{N_{{\rm{tx}}}}R_{\rm{tot}}^j} \right)}^n}} } \right\}
\nonumber\\  & =
\sum\limits_{n = 0}^{{N_{{\rm{th}}}} - 1} {\frac{1}{{n!}}\int_0^\infty  {\exp \left\{ { - N_{{\rm{tx}}}\tau } \right\}{{\left( {N_{{\rm{tx}}}\tau } \right)}^n}} \diff \Pr \left( {R_{\rm{tot}}^j \le \tau } \right)},  
\end{align} 
where
\begin{align}
R_{\rm{tot}}^j = \sum\limits_{x \in {\Phi _a}} {R\left( {\left. {{\Omega _{{r_r}}},j} \right|\left\| {\bf{x}} \right\|} \right)} ,  \label{R_tot_j}
\end{align}
	with ${R\left( {\left. {{\Omega _{{r_r}}},j} \right|\left\| {\bf{x}} \right\|} \right)}$ given in \eqref{R_FA} for the absorbing receiver and in \eqref{R_PS} for the passive receiver.

\bibliographystyle{vancouver}
\bibliography{Ref}

\begin{thebibliography}{10}
\providecommand{\url}[1]{#1}
\csname url@samestyle\endcsname
\providecommand{\newblock}{\relax}
\providecommand{\bibinfo}[2]{#2}
\providecommand{\BIBentrySTDinterwordspacing}{\spaceskip=0pt\relax}
\providecommand{\BIBentryALTinterwordstretchfactor}{4}
\providecommand{\BIBentryALTinterwordspacing}{\spaceskip=\fontdimen2\font plus
\BIBentryALTinterwordstretchfactor\fontdimen3\font minus
  \fontdimen4\font\relax}
\providecommand{\BIBforeignlanguage}[2]{{%
\expandafter\ifx\csname l@#1\endcsname\relax
\typeout{** WARNING: IEEEtran.bst: No hyphenation pattern has been}%
\typeout{** loaded for the language `#1'. Using the pattern for}%
\typeout{** the default language instead.}%
\else
\language=\csname l@#1\endcsname
\fi
#2}}
\providecommand{\BIBdecl}{\relax}
\BIBdecl

\bibitem{yansha2016stochastic}
Y.~Deng, A.~Noel, W.~Guo, A.~Nallanathan, and M.~Elkashlan, ``Stochastic
  geometry model for large-scale molecular communication systems,'' in
  \emph{Proc. IEEE GLOBECOM}, Dec. 2016.

\bibitem{EckfordBook13}
T.~Nakano, A.~Eckford, and T.~Haraguchi, \emph{Molecular communication}.\hskip
  1em plus 0.5em minus 0.4em\relax Cambridge University Press, 2013.

\bibitem{Codling08}
E.~Codling, M.~Plank, and S.~Benhamous, ``Random walk models in biology,''
  \emph{Journal of The Royal Society Interface}, vol.~5, no.~25, pp. 813--834,
  Aug. 2008.

\bibitem{Atkinson09}
S.~Atkingson and P.~Williams, ``Quorum sensing and social networking in the
  microbial world,'' \emph{Journal of The Royal Society Interface}, vol.~6,
  no.~40, pp. 959--978, Aug. 2009.

\bibitem{SMIET2017}
W.~Guo, Y.~Deng, H.~B. Yilmaz, N.~Farsad, M.~Elkashlan, C.~Chae, A.~W. Eckford,
  and A.~Nallanathan, ``{SMIET:} simultaneous molecular information and energy
  transfer,'' \emph{IEEE Wireless Comm.}, 2017.

\bibitem{Llatser13}
I.~Llatser, A.~Cabellos-Aparicio, and M.~Pierobon, ``Detection techniques for
  diffusion-based molecular communication,'' \emph{{IEEE} J. Sel. Areas
  Commun.}, vol.~31, pp. 726--734, Dec. 2013.

\bibitem{Guo15TMBMC}
W.~Guo, C.~Mias, N.~Farsad, and J.~Wu, ``Molecular versus electromagnetic wave
  propagation loss in macro-scale environments,'' \emph{IEEE Trans. Mol. Biol.
  Multi-Scale Commun.}, vol.~1, Mar. 2015.

\bibitem{howard1993radom630}
H.~C. Berg, \emph{Random Walks in Biology}.\hskip 1em plus 0.5em minus
  0.4em\relax Princeton University Press, 1993.

\bibitem{Kuran10}
M.~S. Kuran, H.~B. Yilmaz, T.~Tugcu, and B.~O. Edis, ``Energy model for
  communication via diffusion in nanonetworks,'' \emph{Nano Commun. Netw.},
  vol.~1, no.~2, pp. 86--95, Apr. 2010.

\bibitem{yilmaz2014simulation}
H.~B. Yilmaz and C.-B. Chae, ``Simulation study of molecular communication
  systems with an absorbing receiver: Modulation and {ISI} mitigation
  techniques,'' \emph{Simulat. Modell. Pract. Theory}, vol.~49, pp. 136--150,
  Dec. 2014.

\bibitem{Yansha16}
Y.~Deng, A.~Noel, M.~Elkashlan, A.~Nallanathan, and K.~C. Cheung, ``Modeling
  and simulation of molecular communication systems with a reversible
  adsorption receiver,'' \emph{IEEE Trans. Mol. Biol. Multi-Scale Commun.},
  vol.~1, no.~4, pp. 347--362, Dec. 2015.

\bibitem{tepekule2015ISI}
B.~Tepekule, A.~E. Pusane, H.~B. Yilmaz, C.~B. Chae, and T.~Tugcu, ``{ISI}
  mitigation techniques in molecular communication,'' \emph{IEEE Trans. Mol.
  Biol. Multi-Scale Commun.}, vol.~1, no.~2, pp. 202--216, Jun. 2015.

\bibitem{shih2013code}
P.~J. Shih, C.~H. Lee, P.~C. Yeh, and K.~C. Chen, ``Channel codes for
  reliability enhancement in molecular communication,'' \emph{{IEEE} J. Sel.
  Areas Commun.}, vol.~31, no.~12, pp. 857--867, Dec. 2013.

\bibitem{Douglas12}
S.~M. Douglas, I.~Bachelet, and G.~M. Church, ``A logic-gated nanorobot for
  targeted transport of molecular payloads,'' \emph{Science}, vol. 335, no.
  6070, pp. 831--834, Feb. 2012.

\bibitem{Cavalcanti06}
A.~Cavalcanti, T.~Hogg, B.~Shirinzadeh, and H.~Liaw, ``Nanorobot communication
  techniques: {A} comprehensive tutorial,'' in \emph{Proc. IEEE Int. Conf.
  Control, Autom., Robot., Vis.}, Dec. 2006, pp. 1--6.

\bibitem{adam2014unify}
A.~Noel, K.~C. Cheung, and R.~Schober, ``A unifying model for external noise
  sources and {ISI} in diffusive molecular communication,'' \emph{{IEEE} J.
  Sel. Areas Commun.}, vol.~32, no.~12, pp. 2330--2343, Dec. 2014.

\bibitem{multiplea2015amini}
G.~Aminian, M.~Farahnak-Ghazani, M.~Mirmohseni, M.~Nasiri-Kenari, and F.~Fekri,
  ``On the capacity of point-to-point and multiple-access molecular
  communications with ligand-receptors,'' \emph{{IEEE} Trans. Mol. Biol.
  Multi-Scale Commun.}, vol.~1, no.~4, pp. 331--346, Dec. 2015.

\bibitem{pierobon2014statistical}
M.~Pierobon and I.~F. Akyildiz, ``A statistical--physical model of interference
  in diffusion-based molecular nanonetworks,'' \emph{{IEEE} Trans. Commun.},
  vol.~62, no.~6, pp. 2085--2095, Jun. 2014.

\bibitem{Haenggi12}
M.~Haenggi, \emph{Stochastic geometry for wireless networks}.\hskip 1em plus
  0.5em minus 0.4em\relax Cambridge, UK: Cambridge Uni. Press, 2012.

\bibitem{hasan2007guard}
A.~Hasan and J.~G. Andrews, ``The guard zone in wireless ad hoc networks,''
  \emph{{IEEE} Trans. Wireless Commun.}, vol.~6, no.~3, pp. 897--906, Mar.
  2007.

\bibitem{novlan2013uplink}
T.~D. Novlan, H.~S. Dhillon, and J.~G. Andrews, ``Analytical modeling of uplink
  cellular networks,'' \emph{{IEEE} Trans. Wireless Commun.}, vol.~12, no.~6,
  pp. 2669--2679, Jun. 2013.

\bibitem{yan2016sensor}
Y.~Deng, L.~Wang, M.~Elkashlan, A.~Nallanathan, and R.~K. Mallik, ``Physical
  layer security in three-tier wireless sensor networks: {A} stochastic
  geometry approach,'' \emph{IEEE Trans. Inf. Forensics Security}, to appear
  2016.

\bibitem{yan16hetnet}
Y.~Deng, L.~Wang, M.~Elkashlan, M.~Direnzo, and J.~Yuan, ``Modeling and
  analysis of wireless power transfer in heterogeneous cellular networks,''
  \emph{{IEEE} Trans. Commun.}, no.~99, pp. 1--1, Jul. 2016.

\bibitem{Jeanson11}
S.~Jeanson, J.~Chadoeuf, M.~Madec, S.~Aly, J.~Floury, T.~F. Brocklehurst, and
  S.~Lortal, ``Spatial distribution of bacterial colonies in a model cheese,''
  \emph{Applied and Environmental Microbiology}, vol.~77, no.~4, pp.
  1493--1500, Dec. 2010.

\bibitem{noel2014improving}
A.~Noel, K.~C. Cheung, and R.~Schober, ``Improving receiver performance of
  diffusive molecular communication with enzymes,'' \emph{IEEE Trans.
  Nanobiosci.}, vol.~13, no.~1, pp. 31--43, Mar. 2014.

\bibitem{Yilmaz14}
H.~B. Yilmaz, A.~C. Heren, T.~Tugcu, and C.-B. Chae, ``{Three-dimensional
  channel characteristics for molecular communications with an absorbing
  receiver},'' \emph{IEEE Commun. Lett.}, vol.~18, no.~6, pp. 929--932, Jun.
  2014.

\bibitem{baccelli2009stochastic}
F.~Baccelli and B.~Blaszczyszyn, \emph{Stochastic geometry and wireless
  networks: Volume 1: Theory}.\hskip 1em plus 0.5em minus 0.4em\relax Now
  Publishers Inc, 2009, vol.~1.

\bibitem{schulten2000lectures}
K.~Schulten and I.~Kosztin, \emph{Lectures in Theoretical Biophysics}, 2000.

\bibitem{cussler2009diffusion}
E.~L. Cussler, \emph{Diffusion: Mass Transfer in Fluid Systems}.\hskip 1em plus
  0.5em minus 0.4em\relax Cambridge university press, 2009.

\bibitem{heren2015degradation}
A.~C. Heren, H.~B. Yilmaz, C.~B. Chae, and T.~Tugcu, ``Effect of degradation in
  molecular communication: Impairment or enhancement?'' \emph{{IEEE} Trans.
  Mol. Biol. Multi-Scale Commun.}, vol.~1, no.~2, pp. 217--229, Jun. 2015.

\bibitem{arman2016reactive}
A.~Ahmadzadeh, H.~Arjmandi, A.~Burkovski, and R.~Schober, ``Reactive receiver
  modeling for diffusive molecular communication systems with molecule
  degradation,'' in \emph{Proc. IEEE ICC}, May 2016, pp. 1--7.

\bibitem{noel2014overcome}
A.~Noel, K.~C. Cheung, and R.~Schober, ``Overcoming noise and multiuser
  interference in diffusive molecular communication,'' in \emph{Proc. ACM
  NANOCOM}, May 2014, pp. 1:1--1:9.

\bibitem{nelson2004biological}
P.~Nelson, \emph{Biological Physics: Energy, Information, Life}, updated
  1st~ed.\hskip 1em plus 0.5em minus 0.4em\relax W. H. Freeman and Company,
  2008.

\bibitem{noel2013using}
A.~Noel, K.~C. Cheung, and R.~Schober, ``Using dimensional analysis to assess
  scalability and accuracy in molecular communication,'' in \emph{Proc. IEEE
  ICC MoNaCom}, Jun. 2013, pp. 818--823.

\bibitem{Noel2017}
\BIBentryALTinterwordspacing
A.~Noel, K.~C. Cheung, R.~Schober, D.~Makrakis, and A.~Hafid, ``Simulating with
  {A}c{C}o{RD}: Actor-based communication via reaction–diffusion,''
  \emph{Nano Commun. Netw.}, to appear. [Online]. Available:
  \url{http://doi.org/10.1016/j.nancom.2017.02.002}
\BIBentrySTDinterwordspacing

\bibitem{Haenggi05}
M.~Haenggi, ``{On distances in uniformly random networks},'' \emph{{IEEE}
  Trans. Inf. Theory}, vol.~51, pp. 3584--3586, Oct. 2005.

\bibitem{roman1980formula}
S.~Roman, ``The formula of {F}aa di {B}runo,'' \emph{American Mathematical
  Monthly}, pp. 805--809, 1980.

\end{thebibliography}

\balance
\end{document}